\documentclass[10pt,twocolumn,oneside,letterpaper,journal]{IEEEtran}
\usepackage{cite}
\usepackage{graphicx}
\usepackage{epstopdf}
\usepackage{amsmath}
\usepackage{times}
\usepackage{latexsym}
\usepackage{graphicx}
\usepackage{bm}
\usepackage{bbm}
\usepackage{amssymb}
\usepackage[center]{caption2}
\usepackage{stfloats}
\usepackage{cases}
\usepackage{array}
\usepackage{setspace}
\usepackage{fancyhdr}
\usepackage{fancyhdr}
\usepackage{algorithm}
\usepackage{algorithmic}
\usepackage{textcomp}
\usepackage{wasysym}
\usepackage{balance}
\usepackage{flushend}
\usepackage{url}
\usepackage{xcolor}

\usepackage{booktabs}  

\usepackage{amsmath}

\graphicspath{{figures/}}

\allowdisplaybreaks[4]

\newcaptionstyle{mystyle1}{%
  \centering TABLE  \captiontext \par}
\captionstyle{mystyle1}
\newcaptionstyle{mystyle2}{%
  \captionlabel $.$ \: \captiontext \par}
\captionstyle{mystyle2}
\newcaptionstyle{mystyle3}{%
  \captionlabel $.$ \: \captiontext \par}
\captionstyle{mystyle3}

\newtheorem{proof}{Proof}

\begin{document}

\title{A Federated Reinforcement Learning Method with Quantization for Cooperative Edge Caching in Fog Radio Access Networks}

\author{Yanxiang~Jiang,~\IEEEmembership{Senior~Member,~IEEE},
	Min~Zhang, ~Fu-Chun~Zheng,~\IEEEmembership{Senior~Member,~IEEE},
    Yan~Chen, ~\IEEEmembership{Senior~Member,~IEEE},
	~Mehdi~Bennis,~\IEEEmembership{Fellow,~IEEE}, and Xiaohu~You,~\IEEEmembership{Fellow,~IEEE}
\thanks{{Manuscript received \today. This work was supported in part by the National Key Research and Development Program under Grant 2021YFB2900300, the National Natural Science Foundation of China under grant 61971129, and the Shenzhen Science and Technology Program under Grant KQTD20190929172545139.}}
	\thanks{{Y. Jiang is with the National Mobile Communications Research Laboratory, Southeast University, Nanjing 210096, China, and the School of Electronic and Information Engineering, Harbin Institute of Technology, Shenzhen 518055, China (e-mail: yxjiang@seu.edu.cn).}}
	\thanks{{M. Zhang and X. You are with the National Mobile Communications Research Laboratory, Southeast University, Nanjing 210096, China (e-mail: mzhing@seu.edu.cn, xhyu@seu.edu.cn).}}
	\thanks{{F. Zheng is with the School of Electronic and Information Engineering, Harbin Institute of Technology, Shenzhen 518055, China, and the National Mobile Communications Research Laboratory, Southeast University, Nanjing 210096, China (e-mail: fzheng@ieee.org).}}
    \thanks{{Y. Chen is with the School of Cyber Science and Technology, University of Science and Technology of China, and Key Laboratory of Cyberspace Cultural Content Cognition, Communication and Detection, Ministry of Culture and Tourism (e-mail: eecyan@ustc.edu.cn).}}
	\thanks{{M. Bennis is with the Centre for Wirelesss Communications, University of Oulu, Oulu 90014, Finland (e-mail: mehdi.bennis@oulu.fi).}}
}

\maketitle

\begin{abstract}
In this paper, cooperative edge caching problem is studied in fog radio access networks (F-RANs). Given the non-deterministic polynomial hard (NP-hard) property of the problem, a dueling deep Q network (Dueling DQN) based caching update algorithm is proposed to make an optimal caching decision by learning the dynamic network environment. In order to protect user data privacy and solve the problem of slow convergence of the single deep reinforcement learning (DRL) model training, we propose a federated reinforcement learning method with quantization (FRLQ) to implement cooperative training of models from multiple fog access points (F-APs) in F-RANs. To address the excessive consumption of communications resources caused by model transmission, we prune and quantize the shared DRL models to reduce the number of  model transfer parameters. The communications interval is increased and the communications rounds are reduced by periodical model global aggregation. We analyze the global convergence and computational complexity of our policy. Simulation results verify that our policy has better performance in reducing user request delay and improving cache hit rate compared to benchmark schemes. The proposed policy is also shown to have faster training speed and higher communications efficiency with minimal loss of model accuracy.

\end{abstract}


\section{Introduction}
Currently, mobile data traffic is exploding due to the widespread popularity of mobile services, social networks and resource-intensive applications \cite{1Towards,2RecentAdvances,3Integration}. Various mobile device applications, such as Internet of Things (IoT), Internet of Vehicles (IoV), virtual reality and augmented reality, require higher network throughput and tighter network delay \cite{4AComprehensiveSurvey}. The traditional base station (BS)-centric network architecture can no longer meet these demands, so a new network architecture is needed \cite{5AComprehensiveSurvey}. Fog radio access network (F-RAN) is a promising solution to the congestion problem of cellular network communications. In F-RANs, fog access points (F-APs) are located at the edge of the networks and therefore closer to users than the cloud server \cite{6CooperativeEdgeCaching}. Edge caching technology allows F-APs to store popular contents in the networks \cite{7ContentPopularityPrediction}. Therefore, F-APs can effectively reduce both the load pressure on backhaul links and content delivery delay. Due to the limited communications resources and local storage capacity of F-APs, however, how to cache the most popular contents has become the main challenge of edge caching \cite{8SurveyofFog}.

Several research works have been published on content placement for edge caching. In\cite{9CooperativeEdgeCaching}, by considering the popularity and location factors of contents, a cooperative caching scheme was proposed to cooperatively transmit contents to users in vehicular networks. In \cite{10CostOrientedMobilityAware}, an iterative cache layout algorithm with convex approximation was developed to solve the cache layout problem in device-to-device (D2D) networks.
In \cite{11Gametheoreticapproaches} and \cite{13AMeanField}, game theory was implemented to solve the caching problem by considering different participants (e.g., BS, mobile network operation (MNO), user equipment (UE)) competing with each other to minimize transmission delay and service cost in the networks.
Different from game theory and optimization theory, stochastic geometry schemes realize content caching and replacement by using the corresponding statistical characteristics \cite{14SpatiallyCooperativeCaching,15ContentPlacementfor}.
In \cite{14SpatiallyCooperativeCaching}, the authors developed an analytical framework based on a probabilistic model from stochastic geometry to characterize the performance of the proposed caching policy in terms of cache hit rate.
In \cite{15ContentPlacementfor}, the authors proposed a probabilistic content placement based best tradeoff between content diversity gain and collaborative gain, and stochastic geometry tool was applied to quantify the tradeoff.

On the other hand, learning algorithms have in recent years been used widely to optimize content caching problems by learning key attribute features, such as user request behavior, content popularity, and user mobility distribution.
In \cite{18UserPreferenceLearning}, the authors used the supervised learning approach to design an online content popularity prediction algorithm based on content features and user preference.
In \cite{19Machinelearningbased}, for random content requests, the researchers used $K$-means and $K$-Nearest Neighbor (K-NN) unsupervised algorithms to implement inter-cluster and intra-cluster caching, respectively.
A deep learning-based active caching scheme was proposed in \cite{20DeepCachNet} that collected sparse information about the relevance between content and specific users from a content popularity matrix.

Reinforcement learning (RL) for edge caching has also received significant attention in recent years due to its advantage in decision-making in non-stationary environments. In \cite{21DistributedEdgeCaching} and \cite{22Deepreinforcementlearningfor}, a Markov decision process (MDP) was used to characterize the content request model to describe the temporal and spatial fluctuation of traffic demand in F-RANs.
In \cite{23LearningApproaches} and \cite{24Multi-Agent}, the authors applied a deep reinforcement learning (DRL) framework to implement the dynamic caching decision and optimize the content delivery problem. To overcome the curse of dimensionality of RL and improve the convergence speed, a distributed edge caching algorithm based on deep Q-network (DQN) was considered in \cite{25ContentRecommendation} and \cite{26DuelingDeepQNetwork}. However, most centralized learning algorithms are prone to over-utilization of storage, which can significantly consume network communications resources \cite{27Deepreinforcementlearningfor}. Meanwhile, centralized training requires users to upload their data to the cloud server, and data security of users may therefore be compromised \cite{28VehicularEdgeComputing}. On the other hand, distributed learning approaches do not need data transfer, but require a large amount of cache and action space, which represents a great challenge for edge servers with limited storage and computation resources in F-RANs \cite{29AMulti-AgentDeep}.

Federated learning (FL), as a special machine learning architecture, can protect user privacy by sharing the training model without requiring users to provide their own data.
In \cite{30FederatedDeepReinforcement} and \cite{31CooperativeEdgeCaching}, the authors proposed a cooperative caching strategy based on federated deep reinforcement learning to find the optimal caching policy in F-RANs.
In \cite{32AJointLearning}, the authors studied the training problem of FL algorithm in wireless networks and analyzed the effect of transmit power on the training efficiency of FL.
In \cite{33AdaptiveFederated}, the authors proposed an adaptive asynchronous FL mechanism to change intelligently the number of local update models for achieving global model aggregation.
In \cite{34WhenDeepReinforcement}, the researchers trained a two-timescale RL model by using FL in a distributed manner, in order to achieve real-time and low-overhead computation offloading and resource allocation strategies.

Model transfer in FL, however, can also consume a high amount of communications resources. In the above literature, the communications efficiency problem of FL was not analyzed.
To improve the communications efficiency and reduce the communications cost of FL framework, in \cite{35FederatedLearningStrategies}, the authors proposed two methods to reduce the upstream communications cost: structured updates and sketch updates. In \cite{36FedPAQ} and \cite{37Communication-Efficient}, the authors proposed an FL approach with periodic averaging and quantified communications to address the communications efficiency and scalability challenges. In \cite{38RobustandCommunication-Efficient}, the researchers proposed a sparse-valued compression framework to meet the special requirements of FL environments.

Model computation and communications transmission are two major factors that influence the efficiency of FL. Current intelligent edge devices are equipped with abundant computing resources. The local training of an FL model takes only a few milliseconds, while model transmission takes several seconds, which is influenced by the network bandwidth and the participants distribution of FL. The cross-territory participants make the communications delay of FL grow larger. How to effectively improve the training efficiency and communications efficiency of FL without degrading the model performance has become a key to whether the FL model can be deployed practically.

Inspired by the above discussions, in this paper, we consider a federated reinforcement learning method with quantization (FRLQ) to solve the problem of cooperative caching in F-RANs.
The main contributions are summarized as follows:
\begin{itemize}
\item[1)] We train a deep reinforcement learning model based on dueling DQN in single F-AP of cooperative caching scenario. By learning user request behavior and content popularity, the F-AP can adaptively select caching update actions to make an optimal content caching decision.

\item[2)] We deploy an FL framework leveraging both the cloud and the edge in F-RANs. Through the collaborative training of DRL models from multiple F-APs, the problems of insufficient data and slow convergence of model generated by individual F-AP training alone are solved. At the same time, the FL framework allows user data to always stay in local F-APs  instead of being transferred to the cloud server for model training, which protects user data privacy.

\item[3)] We use periodic updates to increase the communications interval and reduce communications rounds for the FL framework. For the local models to be uploaded, we apply pruning and quantization methods to reduce the number of model parameters. Based on the layer-sensitive transmission method, the model training speed and communications efficiency of the federated reinforcement learning framework are improved.

\item[4)] We analyze the performance of the proposed FRLQ-based cooperative caching policy. We theoretically prove that our policy converges globally and analyze its computational complexity. Simulation results verify the convergence of our proposed policy and show that the transmitted parameters of our policy drop to 60\% with minimal loss of model accuracy compared with the baseline schemes.

\end{itemize}

The rest of this paper is organized as follows.
Section II describes the system model and formulates the caching optimization problem. In Section III, we propose an FRLQ-based cooperative caching policy in F-RANs. In Section IV, simulation results verify the superiority of the proposed cooperative caching policy compared to benchmark schemes. Finally, conclusions are drawn in Section V.

\section{System Model And Problem Formulation}
In this section, we first introduce the network topology of cooperative edge caching in F-RANs. And then, we establish the delay model and the cache replacement model. Finally, we formulate the cooperative caching optimization problem for the system model. For convenience, the notations of some key parameters are presented in Table I.
\subsection{Network Topology}
\begin{figure}
  \centering
  \includegraphics[width=0.44\textwidth]{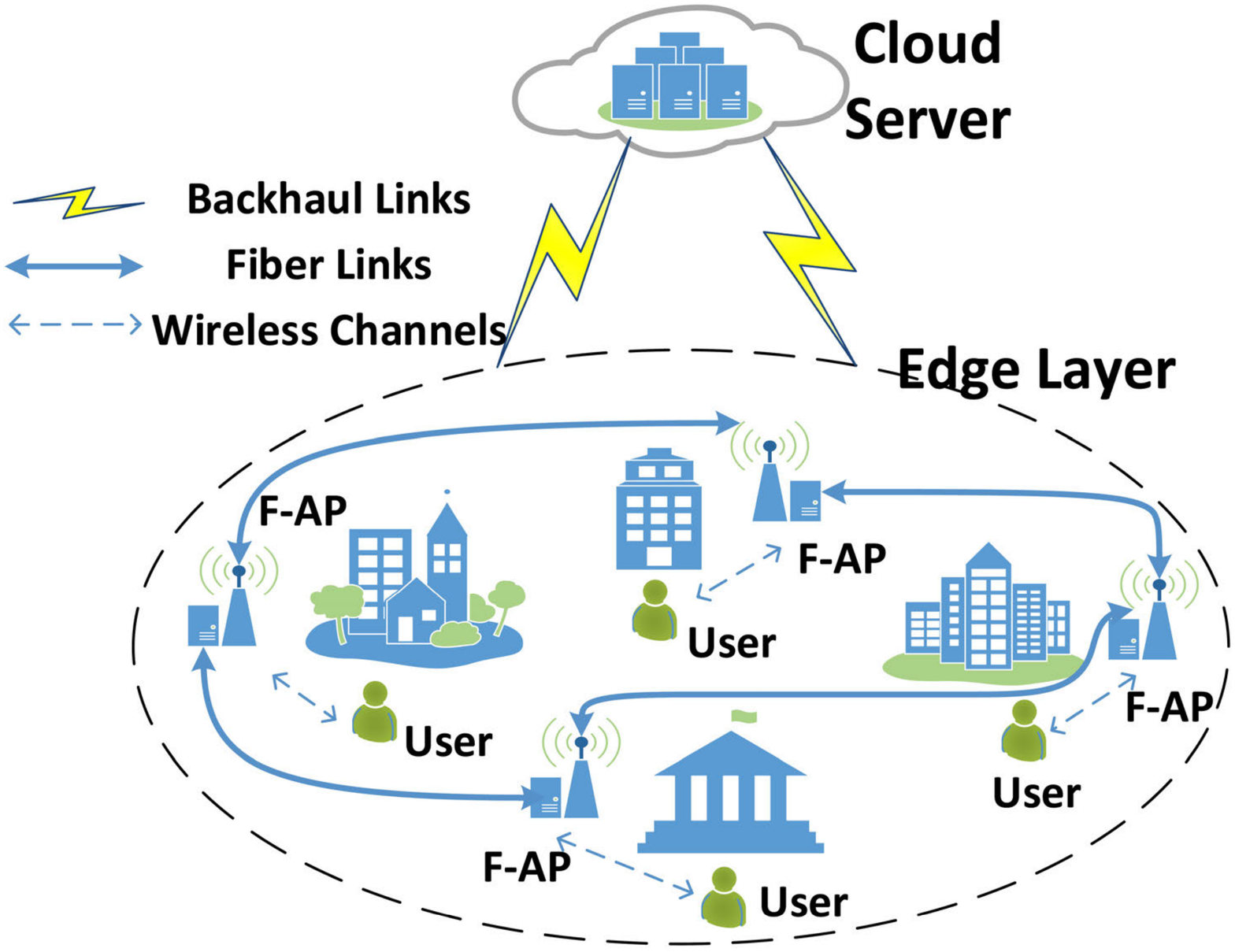}
  \caption{The network topology of cooperative edge caching in F-RANs.}\label{systemmodel}
\end{figure}

The considered network topology of cooperative edge caching in F-RANs is shown in Fig. \ref{systemmodel}, which consists of three parts: the cloud server, the edge layer and the users. The cloud server includes the core network and content providers, which can provide powerful contents and computing services. The edge layer consists of multiple F-APs, which can provide certain computing and caching services for the users. In the network, the cloud server is far away from the edge layer and transmits contents to F-APs through backhaul links. The edge layer serves users who are close to it through wireless channels. The F-APs in the edge layer are connected to each other through fiber links. In the considered cooperative caching scenario, the cloud server serves several F-APs and an F-AP serves multiple users. We denote the set of F-APs as $\mathcal{N}=\{1,2,... ,n,... ,N\}$ and the set of users covered by F-AP $n$ as ${{\mathcal{U}}_{n}}=\{1,2,...,{{u}_{n}},...,{{U}_{n}}\}$. The requested contents of all users are represented as a content library $\mathcal{F}=\{1,2,...,f,...,F\}$, and the size of content $f$ is denoted as $L_f$. We assume that the cloud server can provide all requested contents for users, and each F-AP can only cache a certain amount of requested contents because of the limited cache capacity. The cache capacity of F-AP $n$ is denoted as $C_n$.

\begin{table}[t]\footnotesize \label{keyparameters}
  \centering
  \caption{Key Parameters}
\begin{tabular}{|c|l|}
  \hline
  Notation & Description\\
  \hline
  $N$ & Number of F-APs \\
  \hline
  $U_n$ & Number of users covered by F-AP $n$\\
  \hline
  $F$  & Total number of contents\\
  \hline
  $T$ & Total time slots\\
  \hline
  $L_f$ &  Size of content $f$\\
  \hline
  $C_n$ & Cache capacity of F-AP $n$\\
  \hline
  $P_f$ & Global content popularity of content $f$\\
  \hline
  $p_{n,f}$ & Request probability of content $f$ in F-AP $n$ \\
  \hline
  $\eta$ & The skewness parameter of the M-Zipf distribution\\
  \hline
  $\lambda$ & The plateau parameter of the M-Zipf distribution\\
  \hline
  $d^a$ & Transmission delay between F-APs \\
  \hline
  $d^b$ & Transmission delay between the cloud and the F-AP\\
  \hline
  $d_n^c$ & Transmission delay between F-AP $n$ and user $u_n$\\
  \hline
  $\mathcal{S}_t$ & State space \\
  \hline
  $\mathcal{A}(\mathcal{S}_t)$ & Action space \\
  \hline
  $\mathcal{R}(\mathcal{S}_t,\mathcal{A}(\mathcal{S}_t))$ & Reward function \\
  \hline
  $Q(\mathcal{S}_t,\mathcal{A}(\mathcal{S}_t))$  & Action-value function \\
  \hline
  $V(\mathcal{S}_t)$  & State-value function \\
  \hline
  $\theta$ & Deep neural network parameter \\
  \hline
  $\mathcal{D}$ & Replay memory \\
  \hline
  $G_n$ & Local model of F-AP $n$ \\
  \hline
  $g_{i,n}$ & The $i$th layer of local model in F-AP $n$ \\
  \hline
  $\vartheta_{i,n}$ & Mean change of the parameters in the $i$th layer \\
  \hline
  $\Delta_{i,n} $ & Mean change set of the parameters of all layers \\
  \hline
  $f_n(\theta^n_t)$ & Loss function of local model in F-AP $n$ \\
  \hline
  $f(\theta_t)$ & Loss function of global model \\
  \hline
  $Q(\cdot)$  & Quantization process for the model \\
  \hline
  $Y$ & Total periods of FRLQ\\
  \hline
  $X$ & Total number of local updates in each period \\
  \hline
  \end{tabular}
\end{table}

We divide continuous time into discrete time slots and the set of time slots is denoted as $\mathcal{T}=\{1,2,... ,t,... ,T\}$. The network operates during each time slot. During time slot $t$, the network performs the following three operations. In the first phase, each F-AP will share its content cache status with neighboring F-APs and the cloud server so as to share its locally cached content information. In the second phase, each F-AP will send the content to the user. If the requested content of the user is cached in the local F-AP, the local F-AP will send the content to the user directly. If the requested content is not cached locally, the local F-AP will fetch the content from neighboring F-APs or the cloud server. Finally, each F-AP will update its locally cached contents based on the content request information of its users. We use ${{P}_{f}}$ to denote the global popularity of content $f$, and ${{P}_{f}}$ also represents the probability that content $f$ is requested by all users in the network. The set of global popularity is denoted as $\mathcal{P}=\{{{P}_{1}},{{P}_{2}},...,{{P}_{f}},...,{{P}_{F}}\}$. Let ${{p}_{n,f}}$ denote the probability that content $f$ is requested by the users in F-AP $n$. Here, we have ${{P}_{f}}=\frac{1}{N}\sum\nolimits_{n\in \mathcal{N}}{{{p}_{n,f}}}$ and ${{P}_{f}}$ satisfies the Mandelbrot-Zipf (M-Zipf) distribution \cite{26DuelingDeepQNetwork}:
\begin{equation}\label{contentpopularity}
  {{P}_{f}}=\frac{{{({{I}_{f}}+{{\lambda }})}^{-{{\eta  }}}}}{\sum\nolimits_{j\in \mathcal{F}}{{{({{I}_{j}}+{{\lambda }})}^{-{{\eta  }}}}}},\forall f\in \mathcal{F},
\end{equation}
where ${{I}_{f}}$ denotes the rank of content $f$ in the descending order of global content popularity, and ${{\eta}}$ and ${{\lambda}}$ are the skewness and plateau parameters of the M-Zipf distribution. The higher ${{\eta }}$ is, the lower the popularity is for contents requested by the users.

\subsection{Delay Model}
From Fig. \ref{systemmodel}, the delay of transmitting a content between F-APs is denoted by ${{d}^{a}}$, the delay of transmitting a content between the F-AP and the cloud server is denoted by ${{d}^{b}}$, and the delay of F-AP $n$ transmitting a content to user $u_n$ is denoted by ${{d}^{c}_n}$. For the wireless channel between user $u_n$ and F-AP $n$, we assume that orthogonal frequency division multiple access (OFDMA) is employed where each user occupies a resource block (RB). According to \cite{32AJointLearning}, the downlink rate of F-AP $n$ transmitting a content to its served user $u_n$ can be expressed as:
\begin{equation}\label{downloadrate}
{{c}_{n,{{u}_{n}}}}=\sum\limits_{m=1}^{M}{{{k}_{n,{{u}_{n}},m}}B{{\mathbb{E}}_{{{h}_{n,{{u}_{n}}}}}}({{\log }_{2}}(1+\frac{{{P}_{n}}{{h}_{n,{{u}_{n}}}}}{{{I}_{m}}+B{{N}_{0}}}})),
\end{equation}
where $\boldsymbol{k}_{n,u_n}=[k_{n,u_n,1},... ,k_{n,u_n,M}]^\text{T}$ is an RB assignment vector, $M$ is the total number of RBs, $k_{n,u_n,m}\in{\{0,1}\}$ and $\sum\limits_{m=1}^{M}{{{k}_{n,u_n,m}}}=1$. $k_{n,u_n,m}=1$ means that RB $m$ is assigned to user $u_n$ in the coverage of F-AP $n$, and $k_{n,u_n,m}=0$ otherwise. $B$ is the bandwidth of each RB and $P_n$ is the transmit power of F-AP $n$ for each RB. ${{h}_{n,{{u}_{n}}}}$ is the channel power gain between F-AP $n$ and user $u_n$, which is determined by the distance between F-AP $n$ and user $u_n$. ${N}_{0}$ is the noise power spectral density, and ${I}_{m}$ is the interference power caused by the users that are located in other F-APs using RB $m$.
Then, the delay $d_{n}^{c}$ can be expressed as:
\begin{equation}\label{transmission_delay}
  d_{n}^{c}=\frac{L_f}{{c}_{n,{{u}_{n}}}}.
\end{equation}

\subsection{Cache Replacement Model}
We model the content cache replacement process in F-AP $n$ as an MDP. The state space, action space, and reward function of the MDP are described as follows.

1) State space $\mathcal{S}_{t}$: In the network, the state space consists of the content cache state $\boldsymbol{s}_{t,n}^\text{c}$ of F-AP $n$, the content request state $\boldsymbol{s}_{t,u_n}^\text{r}$ of user $u_n$ and the allocation state $\boldsymbol{k}_{n,u_n}$ of RB. During time slot $t$, the content cache state of F-AP $n$ is denoted by $\boldsymbol{s}_{t,n}^\text{c}=\{s_{t,n,1}^\text{c},... ,s_{t,n,f}^\text{c}\,... ,s_{t,n,F}^\text{c}\},n\in \mathcal{N},f\in \mathcal{F}$, where $s_{t,n,f}^\text{c} \in \{0,1\}$, $s_{t,n,f}^\text{c}=1$ means that F-AP $n$ has cached content $f$, and $s_{t,n,f} ^\text{c}=0$ otherwise. The content request state of user $u_n$ is denoted by $\boldsymbol{s}_{t,u_n}^\text{r}=\{s_{t,u_n,1}^\text{r},... ,s_{t,u_n,f}^\text{r},... ,s_{t,u_n,F}^\text{r}\},u_n\in \mathcal{U}_n, f\in \mathcal{F}$, where $s_{t,u_n,f}^\text{r}\in \{0,1\}$, $s_{t,u_n,f}^\text{r}=1$ indicates that user $u_n$ makes a request for content $f$, and $s_{t,u_n,f}^\text{r}=0$ otherwise. In summary, the state space in our network is represented as:
    \begin{equation}\label{cachestate}
  {\mathcal{S}_{t}}=\{\boldsymbol{s}_{t,n}^\text{c},\boldsymbol{s}_{t,u_n}^\text{r},\boldsymbol{k}_{n,u_n}\}.
    \end{equation}

2) Action space $\mathcal{A}({\mathcal{S}_{t}})$: In order to adapt to the dynamic environment, F-AP $n$ needs to decide which content of the local cache will be replaced and where the user requested content will be processed (local F-AP, neighboring F-APs or the cloud server). $\mathcal{A}({\mathcal{S}_{t}})$ denotes the action space under state $\mathcal{S}_{t}$ and can be expressed as follows:
    \begin{equation}\label{actionspace}
      \mathcal{A}({\mathcal{S}_{t}})=\{a_{t,n,f}^\text{local},\boldsymbol{a}_{t,n,f}^{\text{F-AP}},a_{t,n,f}^\text{cloud},\boldsymbol{a}_{t,n}^\text{R}\}.
    \end{equation}
    The action space $\mathcal{A}({\mathcal{S}^{t}})$ consists of the following four network actions.

    \text{      }a) Local processing action $a_{t,n,f}^\text{local}$: $a_{t,n,f}^\text{local}\in \{0,1\}$ represents whether the requested content $f$ of the user can be processed by the local F-AP $n$, where $a_{t,n,f}^\text{local}=1$ means that the requested content $f$ of the user can be fetched by F-AP $n$, and $a_{t,n,f}^\text{local}=0$ otherwise.

   \text{ }b) F-AP processing action $\boldsymbol{a}_{t,n,f}^{\text{F-AP}}$: When the requested content $f$ of user $u_n$ is not cached in the local F-AP $n$, the request will be routed to neighboring F-APs for processing. $\boldsymbol{a}_{t,n,f}^{\text{F-AP}}=\{a_{t,n,f,1}^{\text{F-AP}},a_{t,n,f,2}^{\text{F-AP}},... ,a_{t,n,f,l}^{\text{F-AP}},... ,a_{t,n,f,N}^{\text{F-AP}}\}$ is the set of F-AP processing actions, where $a_{t,n,f,l}^{\text{F-AP}}\in \{0,1\},l\in \mathcal{N},l\ne n$, $a_{t,n,f,l}^{\text{F-AP}}=1$ means that F-AP $l$ is selected to process the requested content $f$ from user $u_n$, and $a_{t,n,f,l}^{\text{F-AP}}=0$ otherwise.

    \text{ }c) Cloud server processing action $a_{t,n,f}^\text{cloud}$: If user $u_n$ is not able to obtain its requested content $f$ from F-AP $n$ or neighboring F-APs, then F-AP $n$ will fetch the requested content $f$ from the cloud server to serve user $u_n$. $a_{t,n,f}^\text{cloud}\in \{0,1\}$, where $a_{t,n,f}^\text{cloud}=1$ means that the requested content $f$ of the user served by F-AP $n$ is processed by the cloud server, otherwise $a_{t,n,f}^\text{cloud}=0$.

    \text{ }d) Cache update action $\boldsymbol{a}_{t,n}^\text{R}$: Cache update action is denoted by $\boldsymbol{a}_{t,n}^\text{R}=\{a_{t,n,0}^\text{R},a_{t,n,1}^\text{R},... .a_{t,n,j}^\text{R},... ,a_{t,n,F}^\text{R}\}$, where $a_{t,n,j}^\text{R}\in \{0,1\}$ means whether the current cached content $j$ will be replaced by the requested content. $a_{t,n,j}^\text{R}=1$ means that the cached content $j$ will be replaced by the current requested content of user $u_n$, and $a_{t,n,j}^\text{R}=0$ otherwise. In order to ensure real-time response to the user request and to further reduce the cost of cache update in F-AP $n$, cache update will occur when $s_{t,u_n,f}^\text{r}\cap a_{t,n,f}^\text{cloud}\cap a_{t,n,j}^\text{R}=1$.

3) Reward function $\mathcal{R}(\mathcal{S}_{t},\mathcal{A}(\mathcal{S}_{t}))$: There are three ways for users to obtain their requested contents.

\text{ }a) The user obtains its requested content from the local F-AP. If the requested content of the user is cached in the local F-AP, the F-AP will send the content to the user directly. For all user requests, the average delay of transmitting contents to users from the local F-AP can be expressed as:
            \begin{equation}\label{delay1}
  {{D}_{t}^\text{F-U}}=\sum\limits_{f=1}^{F}{\sum\limits_{n=1}^{N}{a_{t,n,f}^\text{local}}}{{p}_{n,f}}d_{n}^{c}.
\end{equation}

 \text{ }b) The user obtains its requested content from neighboring F-APs. If the requested content of the user is not cached in the local F-AP, the local F-AP will fetch the requested content from neighboring F-APs and send it to the user. For all user requests, the average delay of transmitting contents from neighboring F-APs to users can be expressed as:
            \begin{equation}\label{delay2}
  {{D}_{t}^\text{F-F-U}}=\sum\limits_{f=1}^{F}{\sum\limits_{n=1}^{N}{\sum\limits_{l=1}^{N}{a_{t,n,f,l}^\text{F-AP}{{p}_{n,f}}(d_{n}^{c}}}}+{{d}^{a}}).
\end{equation}

\text{ }c) The user obtains its requested content from the cloud server. If the requested content of the user is not cached in the local F-AP and neighboring F-APs, the user will need to fetch its requested content from the cloud server. For all user requests, the average delay of transmitting contents from the cloud server to users can be expressed as:
\begin{equation}\label{delay3}
  {{D}_{t}^\text{C-F-U}}=\sum\limits_{f=1}^{F}{\sum\limits_{n=1}^{N}{a_{t,n,f}^\text{cloud}}}{{p}_{n,f}}(d_{n}^{c}+{{d}^{b}}).
\end{equation}

In summary, during time slot $t$, in order to improve the quality of service (QoS) for users and minimize the average content request delay, the immediate reward function of the network can be expressed as:
        \begin{equation}\label{rewardfunction}
{{\mathcal{R}}}({{\mathcal{S}}_{t}},\mathcal{A}({{\mathcal{S}}_{t}}))=\left\{ \begin{aligned}
  & -{{\zeta }_{1}}{{D}_{t}^{\text{F-U}}},\text{ }\text{       Local Service},\\
 & -({{\zeta }_{2}}{{D}_{t}^{\text{F-F-U}}}+{{\zeta }_{1}}{{D}_{t}^{\text{F-U}}}),\text{     F-APs Service}, \\
 & -({{\zeta }_{3}}{{D}_{t}^{\text{C-F-U}}}+{{\zeta }_{1}}{{D}_{t}^{\text{F-U}}}),\text{     Cloud Service}, \\
\end{aligned}
\right.
\end{equation}
where ${{\zeta }_{1}}$, ${{\zeta }_{2}}$ and ${{\zeta }_{3}}$ denote weight factors, and ${{\zeta }_{1}}+{{\zeta }_{2}}+{{\zeta }_{3}}=1$.

\subsection{Problem Formulation}
Based on (\ref{rewardfunction}), in order to reduce the content request delay of users, our optimization goal is to find a policy ${\pi }$ under the state $\mathcal{S}$ of the network to maximize the long-term cumulative reward expectation, which can be expressed as follows:
\begin{equation}\label{goalproblem}
  {{V}_{{\pi }}}(\mathcal{S})=\underset{T\to +\infty }{\mathop{\lim }}\,E[\sum\limits_{\tau =t}^{T}{{{\gamma }^{\tau -t}}{\mathcal{R}_{\pi}}({\mathcal{S}_{t}},\mathcal{A}({\mathcal{S}_{t}}}))|{\mathcal{S}}],
\end{equation}
where ${\mathcal{R}_{\pi}}({\mathcal{S}_{t}},\mathcal{A}({\mathcal{S}_{t}}))$ denotes the immediate reward generated by the network executing the policy ${\pi }$, and $\gamma$ is a constant between $[0,1]$ denoting the discount factor. For rewards that are further away from the current time slot, they are less important. Our goal is to find the optimal policy ${{\pi }^{*}}({{\mathcal{S}}_{t}},\mathcal{A}({{\mathcal{S}}_{t}}))$ such that
\begin{equation}\label{optimalfunction}
  {{\pi }^{*}}({{\mathcal{S}}_{t}},\mathcal{A}({{\mathcal{S}}_{t}}))=\arg \underset{{{\pi }}}\max {{V}_{{\pi }}}(\mathcal{S}).
\end{equation}
Therefore, the corresponding cooperative caching problem, which requires minimizing the content request delay and finding the optimal caching policy to maximize the long-term cumulative reward, can be expressed as follows:
\begin{align}
  &\arg \underset{{{\pi }}}\max {{V}_{{\pi }}}(\mathcal{S}), \label{goalproblem}\\
 \text{s.t.}&\sum\nolimits_{f\in \mathcal{F}}{L_fa_{t,n,f}^\text{local}}\le C_n, \tag{\ref{goalproblem}a}\nonumber \label{con_a}\\
 & s_{t,n,f}^\text{c},s_{t,u_n,f}^\text{r}\in \{0,1\}, \tag{\ref{goalproblem}b}\nonumber \label{con_b}\\
 &a_{t,n,f}^\text{local},a_{t,n,f,l}^\text{F-AP},a_{t,n,f}^\text{cloud},a_{t,n,j}^\text{R}\in \{0,1\}. \tag{\ref{goalproblem}c}\nonumber \label{con_c}
\end{align}
The constraint (\ref{con_a}) ensures that the size of the cached contents in F-AP $n$ does not exceed its maximum cache capacity $C_n$.

\section{Framework Design of FRLQ}
In this section, in order to address the formulated cache optimization problem, we propose an FRLQ-based cooperative caching policy. First, dueling DQN is used to solve the cache update problem of single F-AP. Then, we deploy an FL framework between the cloud server and the edge layer to achieve collaborative training of local models from multiple F-APs. Subsequently, we quantify the uploaded local model from each F-AP to compress model parameters so as to improve the communications efficiency of the FL framework. Finally, we analyze the global convergence and computational complexity of the proposed policy.
\subsection{Cooperative Edge Caching Policy based on Dueling DQN}
We propose a DRL-based cooperative caching policy to solve the caching optimization problem in (\ref{goalproblem}). Compared with traditional DQN, dueling DQN can solve the reward over-estimation problem and the reward bias problem of DQN, and has strong decision-making ability. Therefore, we employ dueling DQN in each F-AP, which can make the optimal caching decision based on the user request information, the local content caching state, and the content popularity. Let $Q({\mathcal{S}_{t}},\mathcal{A}({\mathcal{S}_{t}}))$ denote the action-value function in dueling DQN. When the local F-AP performs the action $\mathcal{A}({\mathcal{S}_{t}})$ under the state ${\mathcal{S}_{t}}$, it will get the feedback reward.

From the Bellman equation, we can obtain the state-value function as follows:
\begin{equation}\label{statefunction}
  V({\mathcal{S}_{t}})=\max \{{{R}}({\mathcal{S}_{t}},\mathcal{A}({\mathcal{S}_{t}}))+\alpha \sum\limits_{{\mathcal{S}_{t+1}}}{{{P}_{{\mathcal{S}_{t}}{\mathcal{S}_{t+1}}}}V({\mathcal{S}_{t+1}})\}},
\end{equation}
where $\alpha$ is the learning rate, ${P}_{{\mathcal{S}_{t}}{\mathcal{S}_{t+1}}}$ represents the probability of transiting from the state $\mathcal{S}_t$ to the state $\mathcal{S}_{t+1}$. The action-value function $Q({\mathcal{S}_{t}},\mathcal{A}({\mathcal{S}_{t}}))$ is a concrete action representation for the state-value function $V({\mathcal{S}_{t}})$. We choose the optimal action-value function as the state-value function. Then, we have:
\begin{equation}\label{optimalfunction}
  V({\mathcal{S}_{t}})=\underset{\mathcal{A}({\mathcal{S}_{t}})}{\mathop{\max }}\,\{Q({\mathcal{S}_{t}},\mathcal{A}({\mathcal{S}_{t}}))\}.
\end{equation}
The action-value function can be expressed as:
\begin{multline}\label{actionfunction}
  	Q({\mathcal{S}_{t}},\mathcal{A}({{S}_{t}}))={\mathcal{R}}({\mathcal{S}_{t}},\mathcal{A}({\mathcal{S}_{t}}))\\
  +\alpha \sum\limits_{{\mathcal{S}_{t+1}}}{{{P}_{{\mathcal{S}_{t}}{\mathcal{S}_{t+1}}}}}\underset{\mathcal{A}({\mathcal{S}_{t+1}})}{\mathop{\max }}\,Q({\mathcal{S}_{t+1}},\mathcal{A}({\mathcal{S}_{t+1}})).
  \end{multline}
The iterative update formula for the action-value function is:
\begin{multline}\label{actionupdate}
  {{Q}^{*}}({\mathcal{S}_{t}},\mathcal{A}({\mathcal{S}_{t}}))\leftarrow Q({\mathcal{S}_{t}},\mathcal{A}({\mathcal{S}_{t}}))+\alpha ({\mathcal{R}}({\mathcal{S}_{t}},\mathcal{A}({\mathcal{S}_{t}}))\\+
  \gamma \underset{\mathcal{A}({\mathcal{S}_{t+1}})}{\mathop{\max }}\,Q({\mathcal{S}_{t+1}},\mathcal{A}({\mathcal{S}_{t+1}}))-Q({\mathcal{S}_{t}},\mathcal{A}({\mathcal{S}_{t}}))).
\end{multline}

In dueling DQN, a deep neural network $Q({\mathcal{S}_{t}},\mathcal{A}({\mathcal{S}_{t}});\theta )$ is used to model the action-value function
$Q({\mathcal{S}_{t}},\mathcal{A}({\mathcal{S}_{t}}))$, where $\theta$ is the weight parameter of the deep neural network. Then, the loss function of neural network training can be expressed as:
\begin{multline}\label{objectivefunction}
  {{f}}({{\theta }})=\mathbb{E}[({\mathcal{R}}({\mathcal{S}_{t}},\mathcal{A}({\mathcal{S}_{t}}))+\gamma \underset{\mathcal{A}({\mathcal{S}_{t}})}{\mathop{\max }}\,Q({\mathcal{S}_{t+1}},\mathcal{A}({\mathcal{S}_{t+1}});{{\theta }})\\
  -Q({\mathcal{S}_{t}},\mathcal{A}({\mathcal{S}_{t}});{{\theta }}))^{2}],
\end{multline}
where $\gamma$ is the discount factor. The network model is updated by using the gradient descent algorithm. The gradient of the loss function in (\ref{objectivefunction}) is expressed as:
\begin{align}\label{gradient}
  {{\nabla }_{{{\theta }}}}{{f}}({{\theta }})=&\mathbb{E}[({\mathcal{R}({\mathcal{S}_{t}},\mathcal{A}({\mathcal{S}_{t}}))}+\gamma \underset{\mathcal{A}({\mathcal{S}_{t+1}})}{\mathop{\max }}\,Q({\mathcal{S}_{t+1}},\mathcal{A}({\mathcal{S}_{t+1}});{{\theta }}) \notag \\
 &-Q({\mathcal{S}_{t}},\mathcal{A}({\mathcal{S}_{t}});{{\theta }})){{\nabla }_{{{\theta }}}}Q({\mathcal{S}_{t}},\mathcal{A}({\mathcal{S}_{t}});{{\theta }})].
\end{align}

Dueling DQN has an excellent ability of learning and fault tolerance. The reason is that dueling DQN introduces two key factors: target network and experience playback mechanism, which are described in detail as follows.

1) \textbf{Target network}:
Dueling DQN uses two neural networks for learning: the prediction network $Q\left( \mathcal{S}_t, \mathcal{A}(\mathcal{S}_t);{\theta} \right)$ is used to evaluate the Q-value of the current state-action pair; the target network $Q\left( \mathcal{S}_t,\mathcal{A}(\mathcal{S}_t);{\theta}^{-} \right)$ is used to generate the target Q-value, where $\theta^{-}$ is the weight parameter of the target network. The target Q-value can be expressed as:
\begin{equation}\label{targetQ}
  \text{Target Q}=({\mathcal{R}}({\mathcal{S}_{t}},\mathcal{A}({\mathcal{S}_{t}}))+\gamma \underset{\mathcal{A}({\mathcal{S}_{t+1}})}{\mathop{\max}}\,Q({\mathcal{S}_{t+1}},\mathcal{A}({\mathcal{S}_{t+1}});{{\theta}^{-}}).
\end{equation}
Dueling DQN updates the parameter $\theta$ of the prediction network according to the loss function in (\ref{objectivefunction}). After every $M$ steps of iterations, the parameter $\theta $ of the prediction network is copied to the parameter ${{\theta }^{-}}$ of the target network for updating. Double network training can then solve the reward over-estimation problem of DQN.

2) \textbf{Experience playback mechanism}: Dueling DQN introduces an experience playback mechanism to store the experience samples in the replay memory $\mathsf{\mathcal{D}}$ during each time slot. The experience samples can be obtained by the interaction between agents (i.e., F-APs) and the environment. Then, mini-batches of samples are randomly selected from the replay memory $\mathsf{\mathcal{D}}$ for neural network training. Dueling DQN saves a large amount of historical experience samples, and each experience sample is stored in $\left(\mathcal{S}_t,\mathcal{A}(\mathcal{S}_t),\mathcal{R}(\mathcal{S}_t,\mathcal{A}(\mathcal{S}_t)),\mathcal{S}_{t+1} \right)$. It indicates that the learning agent reaches the new state $\mathcal{S}_{t+1}$ after performing the action $\mathcal{A}(\mathcal{S}_t)$ under the state $\mathcal{S}_t$ and obtains the corresponding reward $\mathcal{R}(\mathcal{S}_t,\mathcal{A}(\mathcal{S}_t))$. In the experience playback mechanism, randomly sampling mini-batches for training is helpful to remove the correlation and dependence between samples and reduce the deviation in the evaluation of the value function. It solves the problem of data correlation and non-static distribution to make the model easier to converge.

\begin{figure}
  \centering
  \includegraphics[width=0.48\textwidth]{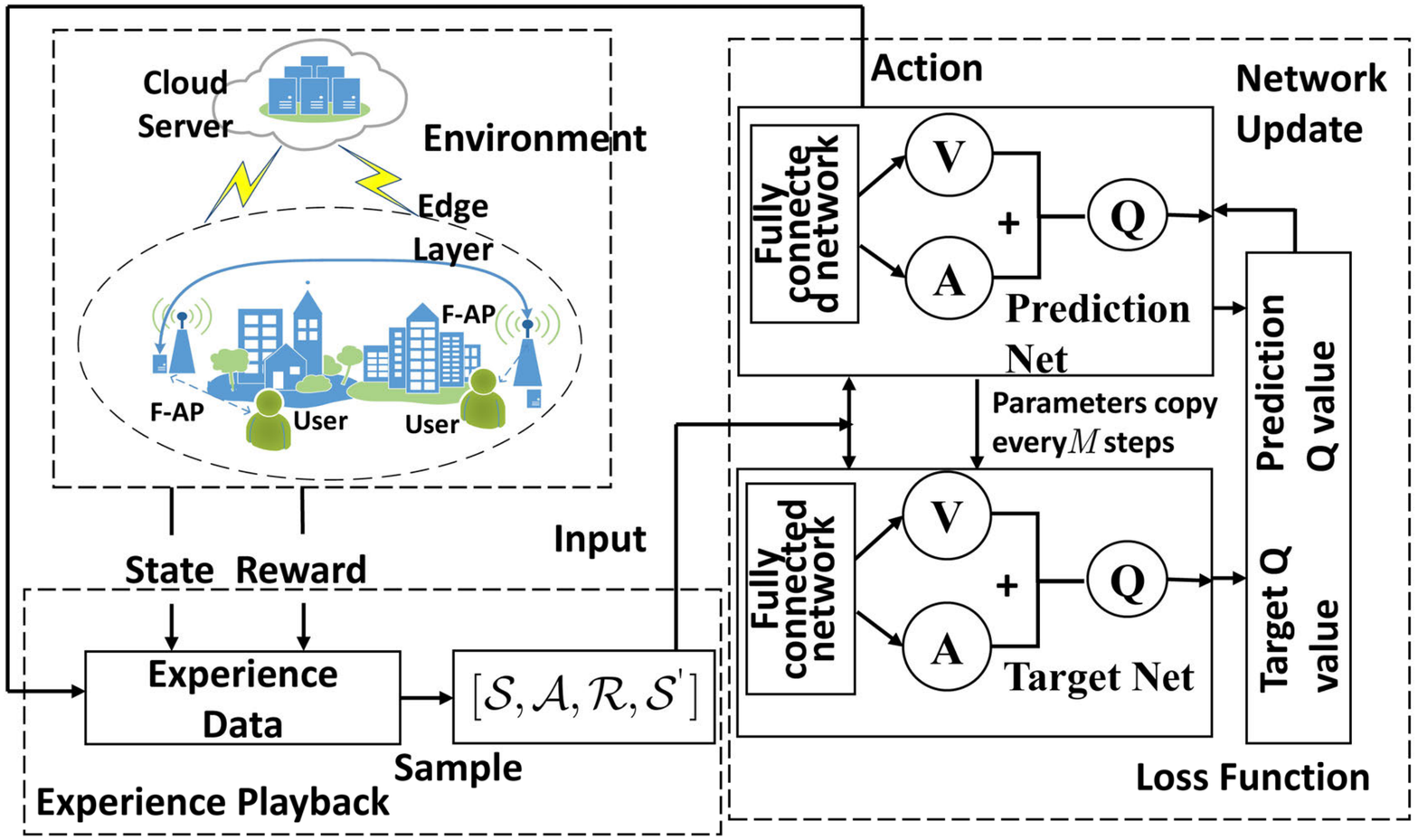}
  \caption{Caching replacement process of dueling DQN.}\label{duelingDQNoperation}
\end{figure}

Dueling DQN replaces the full connected neural network of traditional DQN with the dueling network. The dueling network decomposes the action-value function $Q({\mathcal{S}_{t}},\mathcal{A}({\mathcal{S}_{t}}))$ into a state-value function $V({\mathcal{S} _{t}})$ and an action advantage function $A({\mathcal{S}_{t}},a({\mathcal{S}_{t}}))$. The state-value function $V({\mathcal{S}_{t}})$ is independent of the action, but the action advantage function $A({\mathcal{S}_{t}},a({\mathcal{S}_{t}}))$ is related to the action. $A({\mathcal{S}_{t}},a({\mathcal{S}_{t}}))$ is the average reward estimation of action $a({\mathcal{S}_{t}})$ relative to the state ${\mathcal{S}_{t}}$, which is used to solve the reward bias problem.
Dueling DQN divides the fully connected layer of the neural network into a state-value $V({\mathcal{S}_{t}})$ and an action advantage value $A({\mathcal{S}_{t}},a({\mathcal{S}_{t}}))$. Then, $V({\mathcal{S}_{t}})$ and $A({\mathcal{S}_{t}},a({\mathcal{S}_{t}}))$ are finally merged back into an action-state value $Q({\mathcal{S}_{t}},a({\mathcal{S}_{t}}))$ by a fully connected layer. The operation process of dueling DQN is shown in Fig. \ref{duelingDQNoperation}.

The state-value function $V({\mathcal{S}_{t}})$ in dueling DQN is expressed as:
\begin{equation}\label{duelingstatefunction}
  V({\mathcal{S}_{t}})\cong V({\mathcal{S}_{t}};\theta ,w_1 ).
\end{equation}
The action advantage function $A({\mathcal{S}_{t}},a({\mathcal{S}_{t}}))$ is expressed as:
\begin{equation}\label{actionadvantage}
  A({\mathcal{S}_{t}},a({\mathcal{S}_{t}}))\cong A({\mathcal{S}_{t}},a({\mathcal{S}_{t}});\theta ,w_2 ).
\end{equation}
The action-state value $Q({\mathcal{S}_{t}},a({\mathcal{S}_{t}}))$ is expressed as:
\begin{align}\label{sumfunction}
  Q({\mathcal{S}_{t}},a({\mathcal{S}_{t}}))&\cong Q({\mathcal{S}_{t}},a({\mathcal{S}_{t}});\theta ,w_1 ,w_2 )\notag \\
  &=V({\mathcal{S}_{t}};\theta ,w_1 )+A({\mathcal{S}_{t}},a({\mathcal{S}_{t}});\theta ,w_2 ),
\end{align}
where $w_1 $ and $w_2 $ are the fully connected layer parameters of the two branches.

In dueling DQN, the action advantage is the individual action advantage function minus the average of all action advantage functions in a certain state.
Then, the action-state value $Q({\mathcal{S}_{t}},a({\mathcal{S}_{t}}))$ can be calculated as:
\begin{multline}\label{meanfunction}
  Q({\mathcal{S}_{t}},a({\mathcal{S}_{t}}))=V({\mathcal{S}_{t}};\theta ,w_1 )\\+(A({\mathcal{S}_{t}},a({\mathcal{S}_{t}});\theta ,w_2 )-\frac{1}{\left| A \right|}\sum\limits_{{\boldsymbol{a}'}}{A({\mathcal{S}_{t}},{\boldsymbol{a}'};\theta ,w_2 )}),
\end{multline}
where $\boldsymbol{a}'$ denotes the set of all actions under state $\mathcal{S}_{t}$. (\ref{meanfunction}) ensures that the relative ordering of the advantage functions for each action remains unchanged in a certain state. Therefore, the range of $Q$ value can be narrowed to remove redundant degrees of freedom and improve algorithmic stability. The dueling DQN based content caching replacement process is shown in Algorithm 1.
Each F-AP optimizes the caching policy $\pi$ by maximizing the long-term cumulative reward based on local content popularity and content request information of users.

\subsection{FRLQ: Federated Reinforcement Learning Method with Quantization}
Few problems arise in the above proposed cooperative caching policy based on dueling DQN in single F-AP. Firstly, when the state feature space is small and the training data is limited, it is very challenging to build a high quality content caching policy based on RL. Transfer learning approaches can build a high-performance RL model by transferring experience, parameters, or gradients from learned tasks to new tasks. However, the privacy of data or models in many privacy-aware applications usually does not allow directly transferring data or models from one agent to another. Secondly, it needs to take a long time to train a good RL agent in single F-AP. Meanwhile, the separate training of individual F-AP agents can also cause an additional waste of resources.

To solve the above problems, in this subsection, we propose FRLQ to implement cooperative caching between the cloud server and the edge layer. By sharing the same DRL model between multiple F-APs and the cloud server, it can solve the problem that the samples are too small for DRL model training in single F-AP. At the same time, the proposed policy allows user data to stay in the local F-AP without sharing with the cloud server for training, which protects the users' privacy. However, transmission of massive model parameters brings about the communications efficiency problem for FL. To improve the communications efficiency, on one hand, we prune and quantize the shared DRL model. On the other hand, we periodically perform global model aggregation to increase the communications interval and reduce the communications rounds. Therefore, the proposed FRLQ consists of three modules: (1) network pruning, (2) quantization and weight sharing, and (3) periodical global aggregation. 

\renewcommand{\algorithmicrequire}{\textbf{Initialization:}}
\renewcommand{\algorithmicensure}{\textbf{Iteration:}}
\begin{algorithm}[!t]
\label{algorithm1}
\begin{algorithmic}[1]
\caption{The dueling DQN based content caching replacement algorithm}
\begin{spacing}{0.8}
\REQUIRE
\quad
\\The replay memory $\mathsf{\mathcal{D}}$;
\\The prediction network $Q$ with the weight parameter $\theta $;
\\The target network $\hat{Q}$ with the weight parameter ${{\theta }^{-}}=\theta $;
\ENSURE
\FOR {time slot $t=1,2,...,T$}
\FOR {user ${u}_{n}=1,2,...,{U}_{n}$}
\STATE User ${u}_{n}$ requests content $f$;
\IF {content $f$ has been cached in F-AP $n$ or the neighbouring F-APs}
\STATE {Fetch content $f$ from F-AP $n$ or the neighbouring F-APs};
          \ELSE
             \STATE Fetch content $f$ from the cloud server;
          \IF {the storage of F-AP $n$ is not full or $s_{t,u_n,f}^\text{r}\cap a_{t,n,f}^\text{cloud}\cap a_{t,n,j}^\text{R}=1$}
          \STATE Cache content $f$ in F-AP $n$ directly or replace content $j$;
          \ENDIF
          \ENDIF
          \STATE Observe the state space $\mathcal{S}_t$;
             \STATE According to the $\varepsilon $-greedy method, choose an action ${\mathcal{A}(\mathcal{S}_t)}=\arg \underset{{\mathcal{A}(\mathcal{S}_t)}}{\mathop{\max }}\,Q\left( {\mathcal{S}_t},{\mathcal{A}(\mathcal{S}_t)} \right)$;
             \STATE Execute the action $\mathcal{A}(\mathcal{S}_t)$, and get a new state $\mathcal{S}_{t+1}$;
             \STATE Obtain the reward $\mathcal{R}(\mathcal{S}_t,\mathcal{A}(\mathcal{S}_t))$ according to (\ref{rewardfunction});
             \STATE Save $(\mathcal{S}_t,\mathcal{A}(\mathcal{S}_t),\mathcal{R}(\mathcal{S}_t,\mathcal{A}(\mathcal{S}_t)),\mathcal{S}_{t+1})$ in $\mathsf{\mathcal{D}}$;
             \STATE Set $\mathcal{S}_{t}=\mathcal{S}_{t+1}$;
             \STATE Randomly sample a mini-batch of samples from $\mathsf{\mathcal{D}}$;
\STATE Update ${\theta} $ of the prediction network by using the gradient descent algorithm in (\ref{gradient});
\STATE Reset ${{\theta}^{-}}={\theta} $ every $M$ steps;	
\ENDFOR
\ENDFOR
\end{spacing}
\end{algorithmic}
\vspace{-1.0em}
\end{algorithm}

\subsubsection{\textbf{Network Pruning}}
Training model in an FL scenario is similar to centralized training because there is significant redundancy in the model parameters.
The literature \cite{39PredictingParameters} states that in most deep neural networks, it is possible to achieve performance similar to or even exceeding the original network with only a few weights (e.g., $5{\%}$), while the remaining weights require no need to learn. The transfer of network weights plays an important role in reducing network transmission overhead. Even if an attacker obtains a subset of parameters, it will be difficult for it to use a model inversion attack to invert the original data. 

Inspired by the literature \cite{44han2016deep}, we propose to prune the local models of F-APs to reduce the transmitted parameters. As shown in Fig. \ref{duelingDQNmodel}, we denote the local model in F-AP $n$ as $G_n=\{{{g}_{1,n}},{{g}_{2,n}},...,{g}_{i,n},...,{{g}_{L,n}}\}$, where ${{g}_{i,n}}$ denotes the $i$th layer of the local model in F-AP $n$. After the local training of FL is performed, F-AP $n$ will change the model from ${{G}_{n}}=\{g_{1,n},g_{2,n},...,g_{i,n},...,g_{L,n}\}$ to $G'_{n}=\{g'_{1,n},g'_{2,n},...,g'_{i,n},...,g'_{L,n}\}$. We express the update at the $i$th layer as:
\begin{equation}\label{layer}
  \vartheta _{i,n}=\operatorname{mean}(\left| g'_{i,n}-g_{i,n} \right|).
\end{equation}
 In (\ref{layer}), we first calculate the update size for the weight of each layer, then take the absolute value of the update size, and finally average the absolute values of all the weight updates.
$\vartheta _{i,n}$ is called the sensitivity of the $i$th layer which indicates the mean change of the parameters of the $i$th layer.
 For F-AP $n$, after the local training of the model, the mean change of each layer of the model $\Delta _{i,n}=\{\vartheta _{1,n},\vartheta _{2,n},...,\vartheta _{i,n},...,\vartheta _{L,n}\}$ is calculated according to (\ref{layer}). The change of each layer is ranked from the largest to the smallest. The larger the change of the layer is, the higher the sensitivity of the layer is. Then, by setting the threshold $\chi $ properly, i.e., $\vartheta _{i,n}>\chi$, only a portion of the model update volume $\Delta _{i,n}$ will be uploaded to the cloud server for the global model update. The parameters of the remaining layers will not be uploaded.

\subsubsection{\textbf{Quantization and Weight Sharing}}
Then, we perform the quantization and weight sharing operations on the parameters of the layers to be uploaded. The process is denoted as $Q(\cdot )$.
Quantification and weight sharing further reduce the number of the uploaded parameters by reducing the number of bits required for each parameter upload. We use $K$-means algorithm to cluster the parameters of each layer. All parameters belonging to the same cluster share the same weight value and there is no shared weights across layers. In each layer, we partition $n$ original parameters into $k$ clusters $\mathcal{W}=\{{{w}_{1}},{{w}_{2}},...,{{w}_{k}}\},n>>k$. The goal of clustering is to minimize the within-cluster sum of squares according to \cite{44han2016deep}, which can be expressed as:
\begin{equation}\label{intra-clustersum}
  \arg \underset{\mathcal{W}}{\mathop{\min }}\,{{\sum\limits_{i=1}^{k}{\sum\limits_{w\in \mathcal{W}}{\left| w-{{w}_{i}} \right|}^{2}}}}.
\end{equation}
The position of the centroid is linearly initialized so that the choice of the centroid position will not be affected by the weight distribution. The centroid of $K$-means clustering is the shared weight.

\begin{figure}
  \centering
  \includegraphics[width=0.38\textwidth]{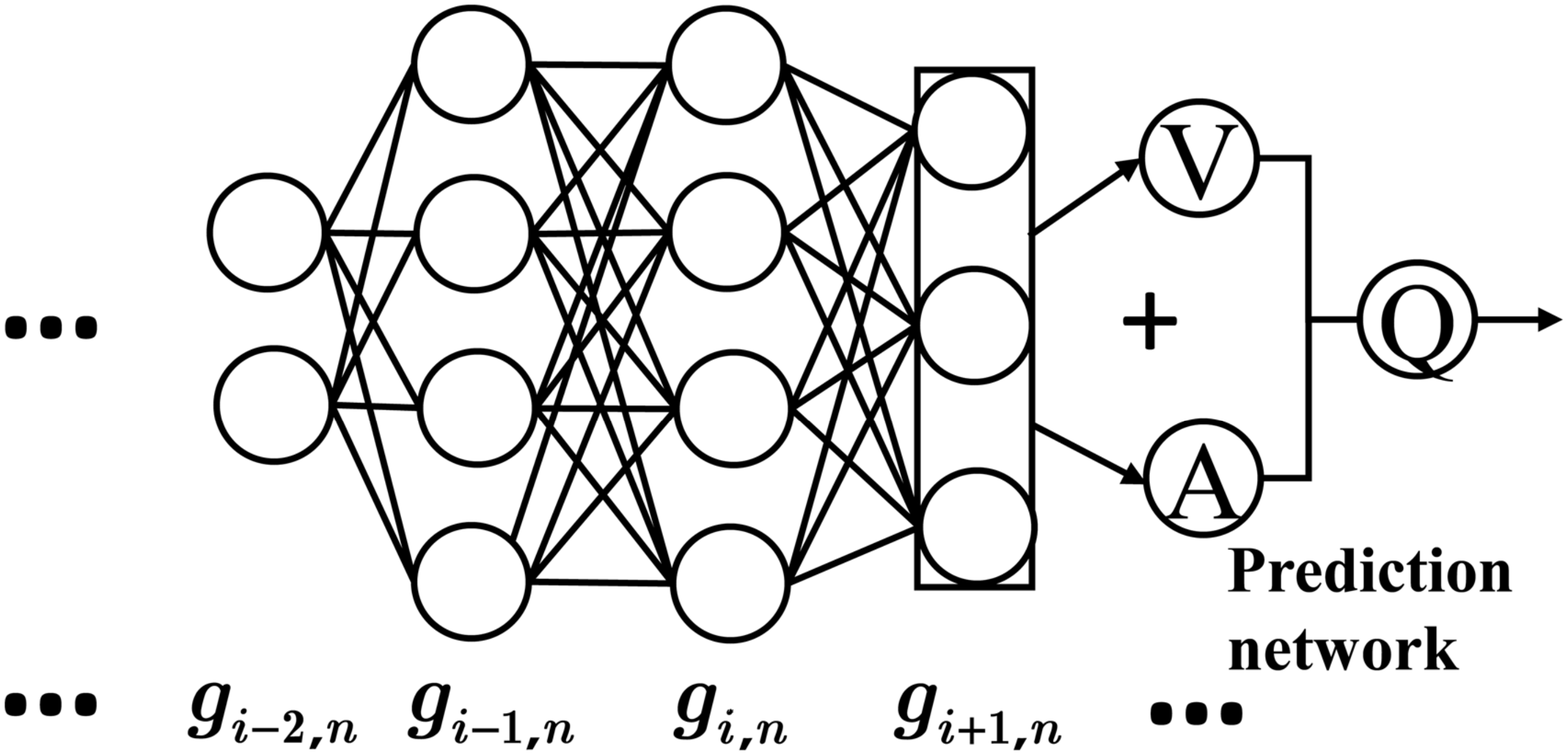}
  \caption{Deep reinforcement learning model in F-AP $n$.}\label{duelingDQNmodel}
\end{figure}

 For $k$ clusters, we only need $\text{log}_{2}(k)$ bits to encode the index. For a network with $n$ connections, it takes $b$ bits for one connection to be uploaded. Limiting $n$ connections to only $k$ shared weights will result in a compression rate of
\begin{equation}\label{compressionrate}
  r=\frac{nb}{n{{\log }_{2}}(k)+kb}.
\end{equation}

For example, we suppose the weights of a single layer neural network with four input units and four output units as shown in Fig. \ref{cluster}. Each square stores the local update of the corresponding weight. Originally we need to spend 16$\times$32 bits to encode 16 weight updates. We perform the quantization and weight sharing process on the single layer neural network. 16 weight updates are compressed into 4 shared weight updates. We only need to spend $16\times\text{log}_{2}(4)$ bits to encode the index and 4$\times$32 bits to encode the four weight updates. Therefore, the data compression rate is 16$\times$32/(16$\times$2+4$\times$32)=3.2.

\begin{figure}
  \centering
  \includegraphics[width=0.48\textwidth]{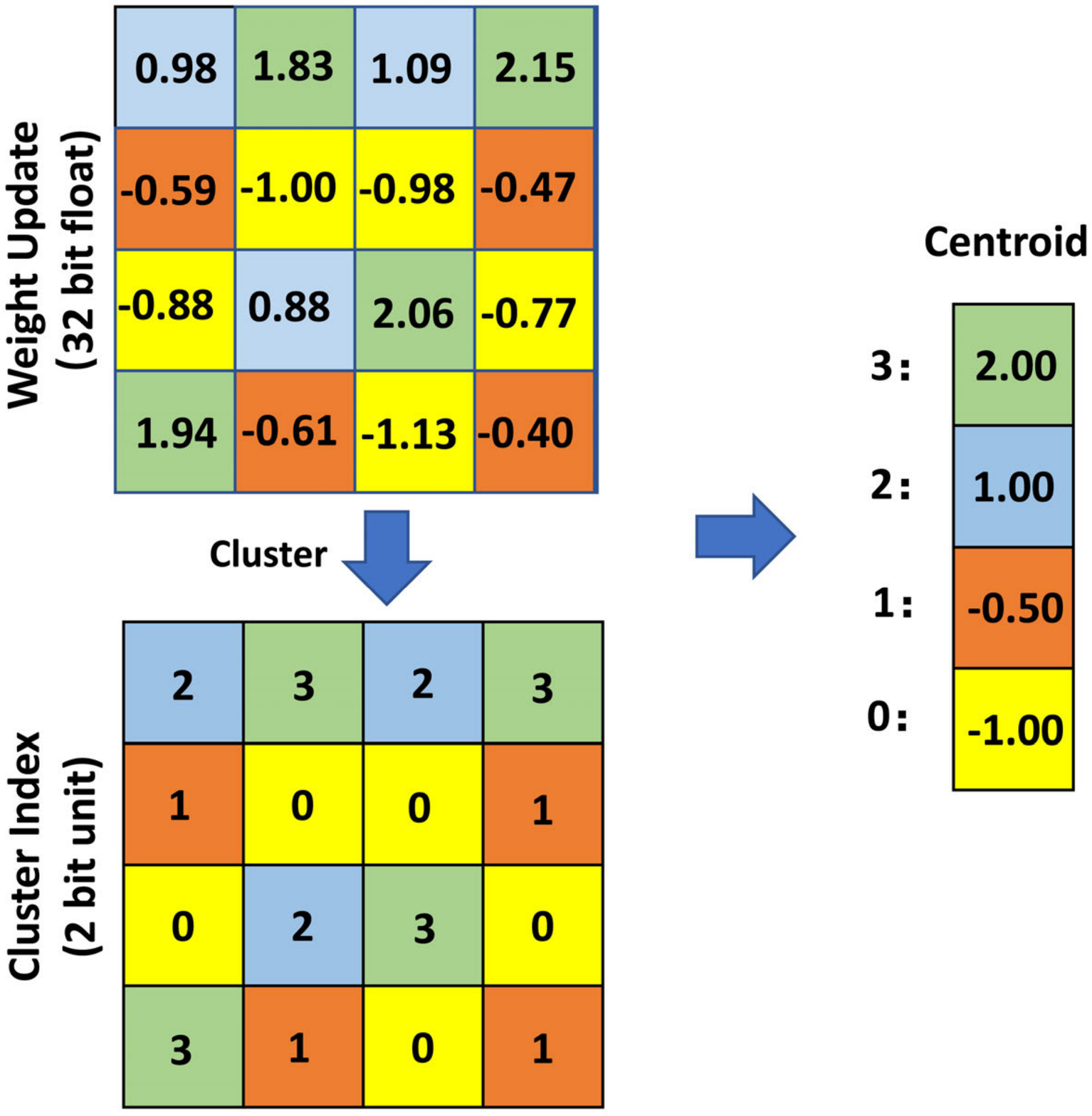}
  \caption{The quantization and weight sharing process.}\label{cluster}
\end{figure}
\subsubsection{\textbf{Periodical Global Aggregation}}
We deploy an FL framework between the cloud server and multiple F-APs. Each F-AP uses its own dataset to train dueling DQN. The loss function of dueling DQN is $f(\theta )$ according to (\ref{objectivefunction}). Therefore, the local learning problem of the FL framework is to find the target parameter $\hat{\theta }_{t}^{n}$ of the loss function${{f}_{n}}(\theta _{t}^n)$ for F-AP $n$, i.e.,
\begin{equation}\label{lossfunction}
  \hat{\theta }_{t}^{n}=\arg \underset{\theta _{t}^{n}}{\mathop{\min }}\,{{f}_{n}}(\theta _{t}^{n}),
\end{equation}
where $\theta _{t}^n$ is the local model parameter of F-AP $n$. The global model can be obtained as follows:
\begin{equation}\label{Theglobalmodel}
  f\left( \theta _{t} \right)=\frac{1}{\sum\nolimits_{n\in \mathcal{N}}{{{\mathcal{D}}_{n}}}}\sum\limits_{n=1}^{N}{{{\mathcal{D}}_{n}}}{{f}_{n}}\left( \theta _{t}^{n} \right),
\end{equation}
where $\theta _{t}$ is global model parameter and ${{\mathcal{D}}_{n}}$ is the local dataset of F-AP $n$. Then, the global learning problem is:
\begin{equation}\label{Theglobalmodelproblem}
 \theta^* = f\left( \theta _{t} \right).
\end{equation}
The learning problem in (\ref{Theglobalmodelproblem}) can be solved by the stochastic gradient descent (SGD) algorithm \cite{30FederatedDeepReinforcement}. For $t=1,2,...,T$, F-AP $n$ updates its local model $\theta _{t}^{n}$ as follows:
\begin{equation}\label{localdataset}
  \theta _{t+1}^{n}=\theta _{t}^{n}-\varphi_t \nabla f_n\left( \theta _{t}^{n} \right),
\end{equation}
where $\varphi_t \ge 0$ is the learning rate. After $T$ iterations, the global model $\theta_{t+1}$ is updated in the cloud server as follows:
\begin{equation}\label{cloudserver}
  \theta _{t+1}=\frac{1}{\sum\nolimits_{n\in \mathcal{N}}{{{\mathcal{D}}_{n}}}}\sum\limits_{n=1}^{N}{{{\mathcal{D}}_{n}}}\theta _{t+1}^{n}.
\end{equation}
The updated global model $\theta _{t+1}$ will be sent to each F-AP for the next round of training. The convergence analysis on the global learning problem in (\ref{Theglobalmodelproblem}) will be shown in the next subsection.

\begin{figure}
  \centering
  \includegraphics[width=0.48\textwidth]{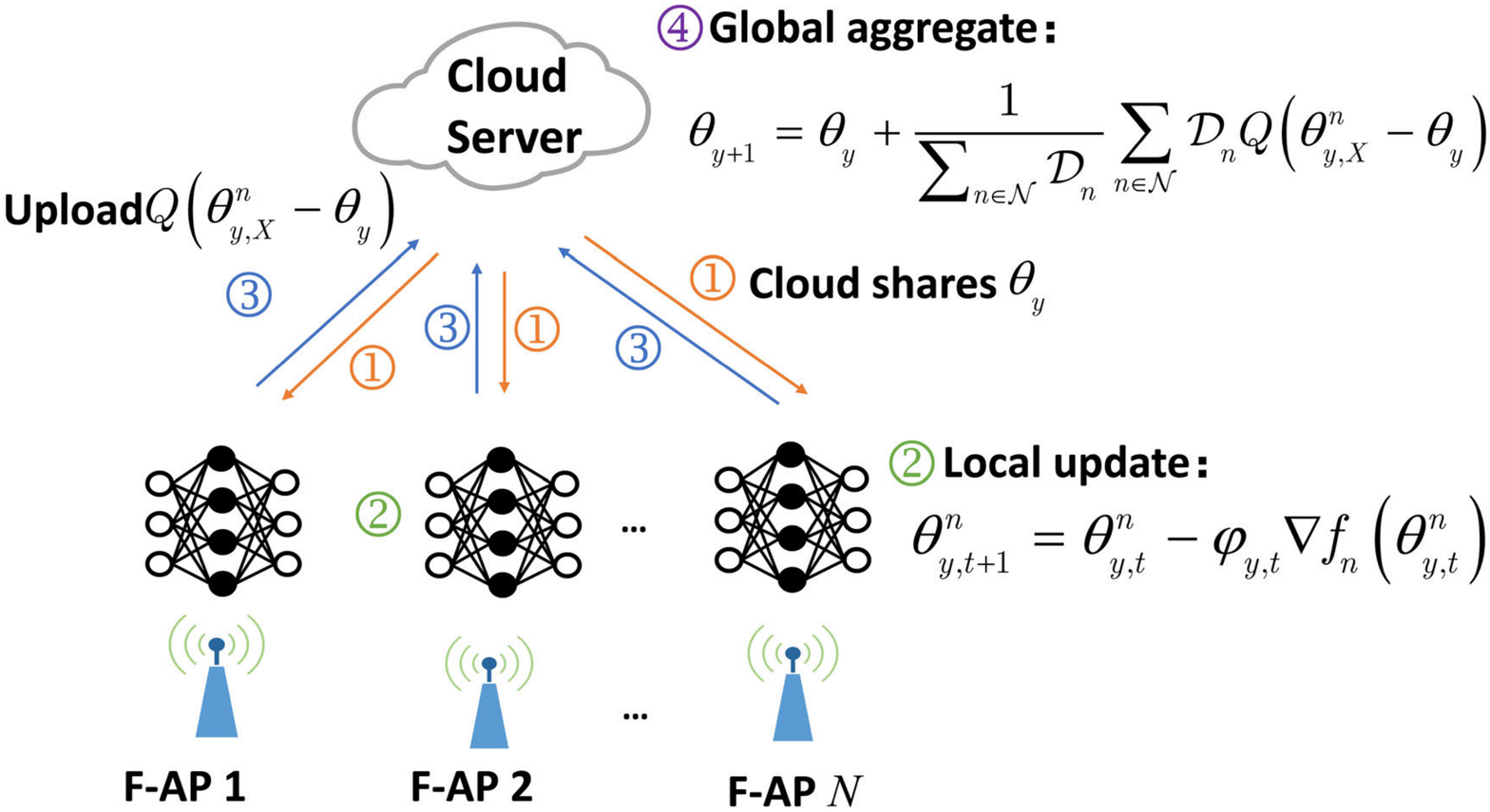}
  \caption{The mechanism of the proposed FRLQ-based cooperative caching policy.}\label{federatedlearning}
\end{figure}

To improve the communications efficiency of the FL framework, we periodically perform global model aggregation to increase the communications interval and reduce the communications rounds. The flow of the proposed FRLQ-based cooperative caching policy is shown in Fig. \ref{federatedlearning}. The proposed policy runs over $Y$ periods. During each period, each F-AP performs $X$ local updates of the DRL model by using its local dataset, so a total $YX$ iterations will be generated. During each period $y=0,1,2,...,Y-1$, each F-AP communicates with the cloud server once and sends the local model parameters to the cloud server. The cloud server aggregates all the received local parameters and then propagates the global model parameters to each F-AP. Specifically, we use $\theta _{y,t}^{n}$ to denote the model parameter of F-AP $n$ in the $t$th iteration of $y$th period. At each local iteration $t=0,1,2,...,X-1$, F-AP $n$ updates the local model according to SGD algorithm as follows:
\begin{equation}\label{stochasticgradientdescent}
  \theta _{y,t+1}^{n}=\theta _{y,t}^{n}-{{\varphi }_{t}}\nabla {{f}_{n}}\left( \theta _{y,t}^{n} \right).
\end{equation}
In this case, stochastic gradient $\nabla {{f}_{n}}$ is calculated by randomly selecting samples from the local dataset ${{\mathcal{D}}_{n}}$.

The proposed policy is formally summarized in Algorithm 2. During every period, the cloud server broadcasts its current model $\theta_y$ to all F-APs. Then, local models of all F-APs use the same initialization $\theta _{y,0}^{n}={{\theta }_{y}}$. After $X$ local updates, F-AP $n$ calculates the parameter update $\theta _{y,X}^{n}-{{\theta }_{y}}$. By considering the layer sensitivity, F-AP $n$ will upload a part of the parameters of the layer by setting the threshold $\chi $. Then, F-AP $n$ executes the process $Q(\cdot )$ to obtain the quantized update $Q\left( \theta _{y,X}^{n}-{{\theta }_{y}} \right)$. Subsequently, F-AP $n$ uploads $Q\left( \theta _{y,X}^{n}-{{\theta }_{y}} \right)$ to the cloud server. After receiving $N$ local quantized updates, the cloud server performs the global aggregation to obtain the global model ${\theta }_{y+1}$ for the next period:
\begin{equation}\label{globalaggregation}
  {{\theta }_{y+1}}={{\theta }_{y}}+\frac{1}{\sum\nolimits_{n\in \mathcal{N}}{{{\mathcal{D}}_{n}}}}\sum\limits_{n\in \mathcal{N}}{{{\mathcal{D}}_{n}}Q\left( \theta _{y,X}^{n}-{{\theta }_{y}} \right)}.
\end{equation}
The model update process will be repeated for $Y$ periods.

\begin{algorithm}[!t]
\label{alg:2}
\begin{algorithmic}[1]
\caption{The FRLQ-based cooperative caching policy}
\REQUIRE
\quad
\\Initialize number of periods $Y$;
\\Initialize step size $\varphi_t $;
\\Initialize local update number $X$;
\ENSURE
\FOR {period $y=1,2,...,Y-1$}
\STATE The cloud server sends $\theta_{y}$ to every F-AP;
\FOR {each F-AP $n \in \mathcal{N}$}
\STATE $\theta _{y,0}^{n}\leftarrow {{\theta }_{y}}$;
\FOR {$t=1,2,...,X-1$}
\STATE F-AP $n$ computes its local update;
\STATE $\theta _{y,t+1}^{n}=\theta _{y,t}^{n}-\varphi_{t} \nabla f_n\left( \theta _{y,t}^{n} \right)$;
\STATE Estimate the convergence according to (\ref{decay});
\ENDFOR
\STATE  Calculate $\vartheta _{i,n}$ and set the threshold $\chi $;
\IF {$\vartheta _{i,n}> \chi $}
\STATE The parameters update of the $i$th layer in F-AP $n$ will be uploaded;
\ELSE
\STATE The parameters update of the $i$th layer in F-AP $n$ will not be uploaded;
\ENDIF
\STATE Perform  $Q(\cdot )$ operation on the uploaded parameter update $\theta _{y,X}^{n}- \theta _y$;
\STATE Send $Q(\theta _{y,X}^{n}- \theta _y)$ to the cloud server;
\ENDFOR
\STATE The cloud server receives $Q(\theta _{y,X}^{n}- \theta _y)$ from each F-AP;
\STATE Compute the global model update:
\STATE ${{\theta }_{y+1}}={{\theta }_{y}}+\frac{1}{\sum\nolimits_{n\in \mathcal{N}}{{{\mathcal{D}}_{n}}}}\sum\limits_{n\in \mathcal{N}}{{{\mathcal{D}}_{n}}Q\left( \theta _{y,X}^{n}-{{\theta }_{y}} \right)}$;
\ENDFOR
\end{algorithmic}
\end{algorithm}

The core idea of the proposed FRLQ-based cooperative caching policy is to combine the cloud server and F-APs to collaboratively train the shared DRL models and accelerate the training process. We use pruning and quantization methods to compress the shared DRL models to reduce the amount of model transmission parameters and lighten the network load. We execute periodic global model aggregation to increase the communications interval and reduce the communications rounds, which further improve the communications efficiency.
\subsection{Performance Analysis}
In this subsection, we will first analyze the global convergence of (\ref{Theglobalmodelproblem}). Subsequently, we will discuss the computational complexity of the proposed FRLQ-based cooperative caching policy.

\subsubsection{\textbf{Global Convergence Analysis}}
In order to analyze the global convergence of (\ref{Theglobalmodelproblem}), we need to make the following assumptions:
\newtheorem{assumption}{\textbf{Assumption}}
\begin{assumption}For all F-APs, $f_n(\theta)$ is $\beta$-smooth and $\mu$-strongly convex, where $f_n(\theta)$ is the objective function of F-AP $n$.
\end{assumption}
\begin{assumption}
 The expected squared norm $G$ of the stochastic gradient is bounded, i.e., $\mathbb{E}{{\left\| \nabla {{f}_{n}}({{\theta }_{t}^n},{{\xi }_{t}^n}) \right\|}^{2}}\le {{G}^{2}}$.
\end{assumption}
\begin{assumption}
The variance $\sigma _{n}^{2}$ of the stochastic gradient for the $t$th iteration of F-AP $n$ is bounded, i.e., $\mathbb{E}{{\left\| \nabla {{f}_{n}}(\theta _{t}^n,\xi _{t}^n)-\nabla {{f}_{n}}(\theta _{t}^n) \right\|}^{2}}\le \sigma _{n}^{2}$, where $\xi _{t}^n$ is the data sampled from the local dataset of F-AP $n$ randomly.
\end{assumption}

We introduce the variable $v_{t+1}^n$ as follows,
\begin{equation}\label{bianliang1}
  v_{t+1}^n=\theta_{t}^n-{{\varphi }_{t}}\nabla {{f}_{n}}(\theta _{t}^n,\xi _{t}^n),
\end{equation}
where $v_{t+1}^n$ is just an intermediate variable and not $\theta _{t+1}^n$ at the next moment, because we have the possibility to execute global aggregation. Then, we have:
\begin{equation}\label{bianliang2}
  \theta _{t+1}^n=\left\{ \begin{aligned}
  & v_{t+1}^n,\text{for}\text{  }t+1\notin {\mathcal{I}_{X}}, \\
 &  \sum\limits_{n=1}^{N}{{{p}_{n}}{{v}_{t+1}^n}},\text{for }t+1\in {\mathcal{I}_{X}}, \\
\end{aligned} \right.
\end{equation}
where $p_n$ is the proportion of the objective function of F-AP $n$ in the global objective function, which is affected by the size of the local dataset of F-AP $n$. $\mathcal{I}_{X}$ is the set of integer multiples of $X$, i.e., global parameter aggregation occurs when $t+1\in {\mathcal{I}_{X}}$.

We can then get the following four variables: ${{\bar{v}}_{t}}=\sum\limits_{n=1}^{N}{{{p}_{n}}v_{t}^n}$, ${{\bar{\theta}}_{t}}=\sum\limits_{n=1}^{N}{{{p}_{n}}\theta _{t}^n}$, ${{g}_{t}}=\sum\limits_{n=1}^{N}{{{p}_{n}}\nabla {{f}_{n}}(\theta _{t}^n,\xi _{t}^n})$, ${{\bar{g}}_{t}}=\sum\limits_{n=1}^{N}{{{p}_{n}}}\nabla {{f}_{n}}(\theta _{t}^n)$. Here, $\theta_{t}^n$ represents the model parameter when the communications transmission can be performed during $t\in \mathcal{I}_X$. When $t\notin \mathcal{I}_X$, communications exchange cannot take place, and the model parameter is denoted by $v_{t}^n$. Due to global participation, ${{\bar{v}}_{t}}={{\bar{\theta }}_{t}}$ for all $t\in \mathcal{I}_X$, and ${{\bar{v}}_{t+1}}={{\bar{\theta }}_{t}}-{{\varphi }_{t}}{{g}_{t}}$. Then, we have the following theorem.

\newtheorem{theorem}{\textbf{Theorem}}
\begin{theorem}
Let Assumption 1, Assumption 2 and Assumption 3 hold. If ${\varphi }_{t}$ is non-increasing with time (i.e., the learning rate is decaying) and ${{\varphi }_{t}}\le \frac{1}{4\beta}$, the proposed FRLQ-based cooperative caching policy follows that:
\begin{equation}\label{globaloptium}
\mathbb{E}{{\left\| {{{\bar{v }}}_{t+1}}-{{\theta }^{*}} \right\|}^2}\le (1-\mu {{\varphi }_{t}})\mathbb{E}{{\|{{\bar{\theta }}_{t}}-{{\theta }^{*}}\|}}^2 +{\varphi }_{t}^{2}H,
\end{equation}
where $\theta^*$ is the global optimal parameter, $H=\sum{p_{n}^{2}\sigma _{n}^{2}}+6\beta\Phi +8{{(X-1)}^{2}}{{G}^{2}}$, $\Phi  ={{f}}(\theta^* )-\sum\limits_{n=1}^{N}{{{p}_{n}}f_{n}^{*}(\theta )}$, ${{f}}(\theta^* )$ is the global optimal function value, and $f_{n}^{*}(\theta )$ is the optimal function value obtained by optimizing F-AP $n$ alone.
\end{theorem}

\begin{proof}
Please see Appendix A.
\end{proof}

 Theorem 1 indicates that the proposed FRLQ-based cooperative caching policy is bounded. We set ${{\Delta }_{t}}=\left\| {{{\bar{\theta }_t}}}-{{\theta }^{*}} \right \|$ and further have:
\begin{equation}\label{conclusion}
  {{\Delta }_{t+1}}\le (1-\mu {{\varphi }_{t}}){{\Delta }_{t}}+{\varphi }_{t}^{2}H.
\end{equation}
To verify the global convergence of the proposed cooperative caching policy, we need to prove that ${{\Delta }_{t}}$ decreases gradually over time. Then, we have the following theorem.

\begin{theorem}
 Assume that ${{\varphi }_{t}}=\frac{b}{t+a}$, where $b>\frac{1}{\mu }$ and $a>0$, such that ${{\varphi }_{1}}\le \min \{\frac{1}{\mu },\frac{1}{4\beta }\}$ and ${{\varphi }_{t}}\le 2{{\varphi }_{t+X}}$. Then, ${{\Delta }_{t}}$ decreases gradually over time, i.e.,
\begin{equation}\label{decay}
  {{\Delta }_{t}}\le \frac{\rho }{a+t},
\end{equation}
where $\rho =\max \{\frac{{{b}^{2}}H}{b\mu -1},(a+1){{\Delta }_{1}}\}$.
\end{theorem}
\begin{proof}
Please see Appendix B.
\end{proof}

Theorem 2 indicates the upper bound of the distance between the iterative parameter ${\bar{\theta}}_{t}$ and the global optimal parameter $\theta^*$ is gradually reduced. Therefore, the proposed FRLQ-based cooperative caching policy is able to converge to the global optimum. In Algorithm 2, we can evaluate the convergence of the proposed policy (line 8) according to Theorem 2 to determine whether the algorithm will stop the training.
\subsubsection{\textbf{Computational Complexity Analysis}}
For the dueling DQN based content caching replacement algorithm, we need to respond to the content requests of users during $T$ time slots. The time complexity of the algorithm is $\mathcal{O}(Tz)$, where $z$ denotes the total number of iterations. For the experience playback mechanism in dueling DQN, it is assumed that $E$ experience samples are stored, so the space complexity of the algorithm is $\mathcal{O}(E)$. Similar to the related studies in \cite{30FederatedDeepReinforcement}, the proposed dueling DQN-based content cache placement algorithm is lightweight and easy to deploy in each F-AP.

For the FRLQ-based cooperative caching policy, the caching policy is repeated over $Y$ periods, and $N$ F-APs perform $X$ local model updates in each period. We consider that the global aggregation of FL uses the synchronous update algorithm, so the time complexity of the algorithm is $\mathcal{O}(YX)$. In the model quantization process, we use the $K$-means algorithm to cluster and compress the parameters of the uploaded models. The time complexity of the $K$-means algorithm is $\mathcal{O}({{i}_{K\text{-means}}}kW)$, where $k$ denotes the number of clusters, ${{i}_{K\text{-means}}}$ denotes the maximum iteration number of the clustering algorithm, and $W$ denotes the number of parameters to be clustered. Accordingly, the time complexity of the FRLQ-based cooperative caching policy is $\mathcal{O}({{i}_{K\text{-means}}}knYX)$. Here, $X$ and $Y$ are preset constants before training, and $n$ is the number of weights in each layer of the neural network, which is also known. Therefore, the time complexity of the proposed FRLQ-based cooperative caching policy is lower and on par with the traditional $K$-means algorithm.
\section{Simulation Results}
The performance of the proposed FRLQ-based cooperative caching policy is evaluated via simulations. We consider a cloud server whose coverage is a circular area with a radius of 500 m.  We assume that the content library contains 1000 contents and the size of each content $L_f$ is 1 MB. The settings of other parameters are given in Table II.
Dueling DQN consists of a fully connected feedforward neural network in the hidden layer, which is used to construct the prediction network $Q$ and the target network $\hat{Q}$.

We adopt the following caching schemes as benchmarks, which are described in detail as follows.

1) \textbf{Least recently used (LRU)}: The stored content that is least recently requested will be replaced by the new content.

2) \textbf{Least frequently used (LFU)}: The stored content with the least request frequency will be replaced.

3) \textbf{Centralized DRL}: The dynamic content replacement is operated in the cloud server by deploying the centralized DRL algorithm \cite{40wang2016dueling}.

4) \textbf{Federated deep reinforcement learning (FRL)}: This scheme employs the FL framework without pruning and quantizing the transmission model \cite{30FederatedDeepReinforcement}.

The goal of the paper is to find the optimal content caching policy by minimizing the content request delay of users under the dynamic network environment and cache capacity constraint of F-APs. We perform pruning and quantization on the shared model, which ensures that our model is lightweight and easy to deploy. First, we simulate the convergence of the proposed policy and the variation of model accuracy after pruning and quantization. Then, in order to evaluate the network performance of the proposed policy, we do the simulations by taking average request delay and cache hit rate as dependent variables, the cache size of F-AP $C_n$, the skewness factor $\eta$ and the number of F-APs $N$ as independent variables, respectively. The cache hit rate of F-AP $n$ during time slot $t$ can be calculated as follows:
\begin{equation}
{{h}_{t,n}}=1-\sum\limits_{l\in \mathcal{N}}{\sum\limits_{f\in \mathcal{F}}{\left( {{a}_{t,f,n,l}^\text{F-AP}}+{{a}_{t,f,n}^{\text{cloud}}} \right)}}{{p}_{n,f}}.
\end{equation}
The cache hit rate indicates the probability that the requested content of users can be obtained from the local cache of the F-AP serving them. Compared with four baseline schemes, we can see the superiority of the proposed policy.

\begin{table}[t]\footnotesize \label{Table2}
  \caption{Parameter Values}
  \centering
  \vspace{1em}
\begin{tabular}{|c|c|c|}
  \hline
  Parameter & Value & Description \\
  \hline
  ${\eta}$ & 0.8 & the skewness factor \\
  \hline
  ${\lambda}$ & 0.1 & the plateau factor \\
  \hline
  $B$ & 20 MHz &  RB bandwidth \\
  \hline
  $P_n$ & 1 W &  transmit power of F-AP $n$  for each RB \\
  \hline
  $N_0$ & -174 dBm/Hz & noise power spectral density \\
  \hline
  $N$  &  10  & the number of F-APs \\
  \hline
  $d_a$ & 2 ms & the delay of F-APs cooperation \\
  \hline
  $d_b$ & 10 ms & the delay of cloud-F-AP \\
  \hline
  ${{\varsigma }_{1}}$, ${{\varsigma }_{2}}$, ${{\varsigma }_{3}}$ & 0.1, 0.2, 0.7 & the hyper parameters \\
  \hline
  ${{\alpha }}$ & 0.001 & the learning rate\\
  \hline
  $\gamma $ & 0.9 & the discount factor\\
  \hline
\end{tabular}
\end{table}

Fig. \ref{lossone} compares the learning curves of the proposed FRLQ-based caching policy and the traditional FRL-based caching policy for the model training loss. It can be seen that the loss curve decreases rapidly as the number of training steps increases. After enough training steps, the loss converges to a stable state. Compared with the traditional FRL-based caching policy, the proposed caching policy is less stable at the beginning of model training. The model parameters change a lot at the beginning of model training, and the proposed policy further quantizes and compresses the model to make it lighter, so the model converges slowly. 
We can also see that the proposed policy converges when the training steps reach 400 or more.

Fig. \ref{parameternumber} shows the ratio of the number of uploaded parameters to the number of parameters of the FRL-based policy versus training periods for different policies. By adjusting the threshold  $\chi $, we set the percentage of uploaded layers to 80\% and 90\% for simulations respectively. When 90\% of the layers participate in the global model update, the number of parameters after pruning and quantization only accounts for $50\%\sim60\%$ of the parameters of the FRL-based policy. When 80\% of the layers participate in the global model update, the number of parameters after pruning and quantization is further reduced to only $30\%\sim40\%$.
As the period increases, the uploaded model parameters decrease slowly. The reason is that as the period increases, the model gradually converges. The weight update amount of each layer of the model becomes smaller and smaller, and the uploaded parameters participating in the global model update also decrease.

Fig. \ref{losstwo} depicts how the training loss of the model changes with the number of periods.
As the period increases, the training loss of the model decreases gradually.
The convergence of the proposed policy based on FRLQ (90\%) is significantly better than that of the proposed policy based on FRLQ (80\%). The reason is that the more the model parameters are uploaded, the faster the convergence speed will be.

\begin{figure}
  \centering
  \includegraphics[width=0.44\textwidth]{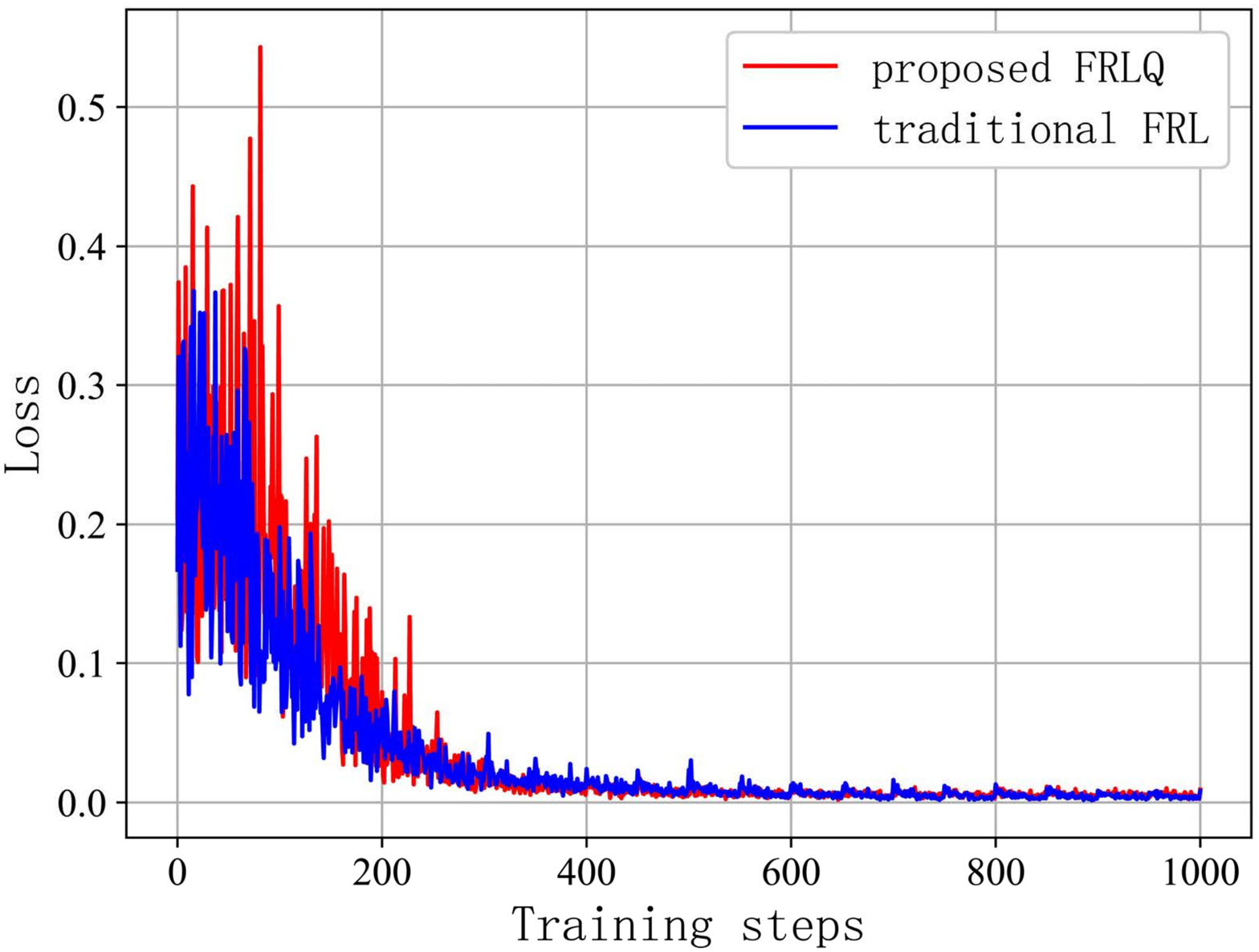}
  \caption{Loss between Q-values and target values for varying training steps.}\label{lossone}
\end{figure}

Fig. \ref{hitrate_timeslot} shows the performance comparison between the proposed policy and the baseline schemes in terms of cache hit rate. We can observe that cache hit rate increases with time slots. When $t<1000$, the cache hit rates of the proposed policy, the traditional FRL-based policy and the centralized DRL-based policy are smaller than that of LFU and LRU. The reason is that the RL-based caching policies have fewer samples to learn in the replay memory at the beginning of the training process. As the time slot increases and the learning samples accumulate, the performance of the RL-based policies is significantly better than that of LFU and LRU. We can also see that the cache hit rate of the proposed policy grows significantly faster than that of the FRL-based caching policy. The reason is that the proposed policy only needs 60\% of the latter on the number of uploaded parameters, which can be seen from the Fig. \ref{parameternumber}. The proposed policy significantly reduces the network load and makes the network cache update faster, and the cache hit rate of F-AP also increases rapidly.


\begin{figure}
  \centering
  \includegraphics[width=0.42\textwidth]{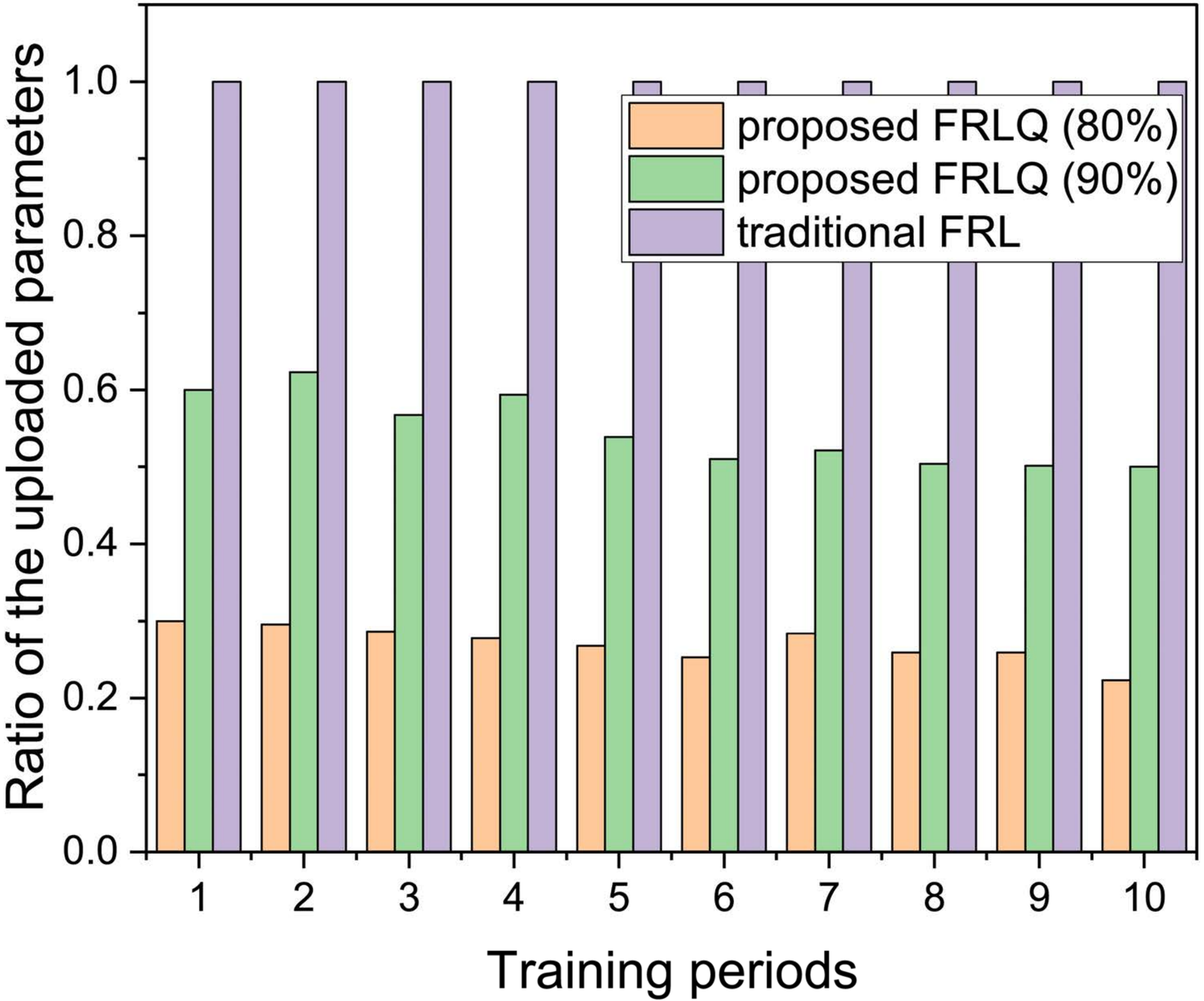}
  \caption{Ratio of uploaded parameters for varying training periods.}\label{parameternumber}
\end{figure}

\begin{figure}
  \centering
  \includegraphics[width=0.42\textwidth]{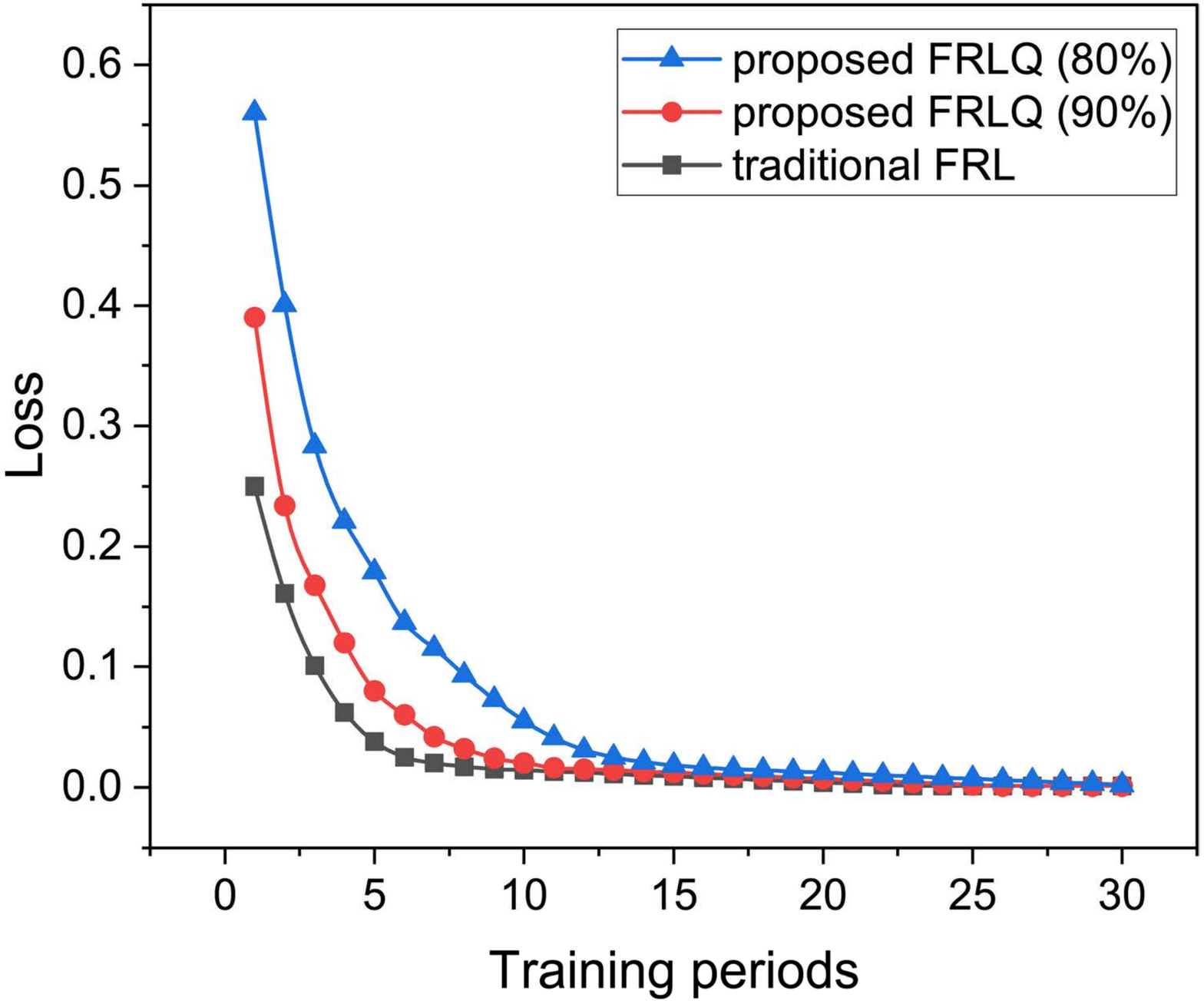}
  \caption{Loss between Q-values and target values for varying training periods.}\label{losstwo}
\end{figure}

Fig. \ref{hitrate_cachesize} shows the variation of the cache hit rate with cache size of single F-AP for the proposed policy and the baseline caching schemes. We can observe that the cache hit rate increases with the cache size of single F-AP. The reason is that more contents can be cached locally, so the probability that the content requests can be answered locally also increases.
We can also observe that although the number of model parameters of the proposed policy is significantly reduced, the cache hit rate is however significantly higher than that of LFU and LRU, and only slightly lower than that of the FRL-based caching policy and the centralized DRL-based caching policy.

In Fig. \ref{delay_alpha}, we show how the average request delay varies with the skewness factor $\eta $ for different caching policies. It can be observed that the average request delay decreases as the skewness factor $\eta $ increases. A larger skewness factor $\eta $ means that more popular contents can be cached in local F-APs, therefore the average request delay decreases correspondingly.

In Fig. \ref{delay_number}, we plot the curves of the average request delay as the number of F-APs varies. It can be seen that the average request delay increases with the number of F-APs. A larger number of F-APs means more F-APs with different local content popularity, and the variety of contents requested by users increases with the number of F-APs. However, the local cache capacity of F-APs remains unchanged, so the average request delay increases. Meanwhile, since all F-APs in the network share the limited RB resources, the F-AP's transmit power $P_n$ remains constant, but the signal interference power $I_m$ increases rapidly with the number of F-APs. Therefore, the rate of F-APs transmitting contents to users decreases, and the average request delay increases correspondingly.

\begin{figure}
  \centering
  \includegraphics[width=0.42\textwidth]{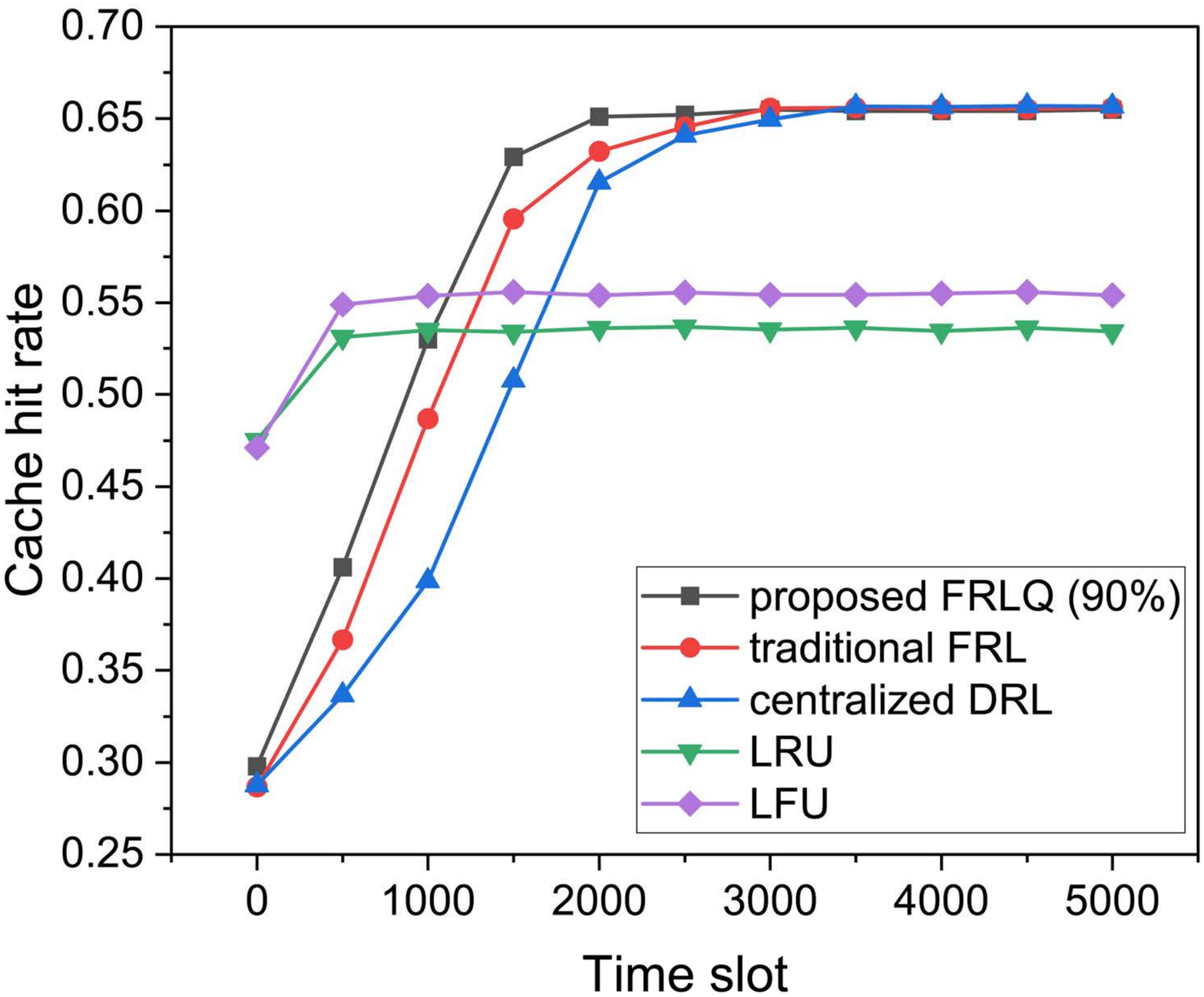}
  \caption{Cache hit rate versus time slot for different caching policies.}\label{hitrate_timeslot}
\end{figure}

\begin{figure}
  \centering
  \includegraphics[width=0.42\textwidth]{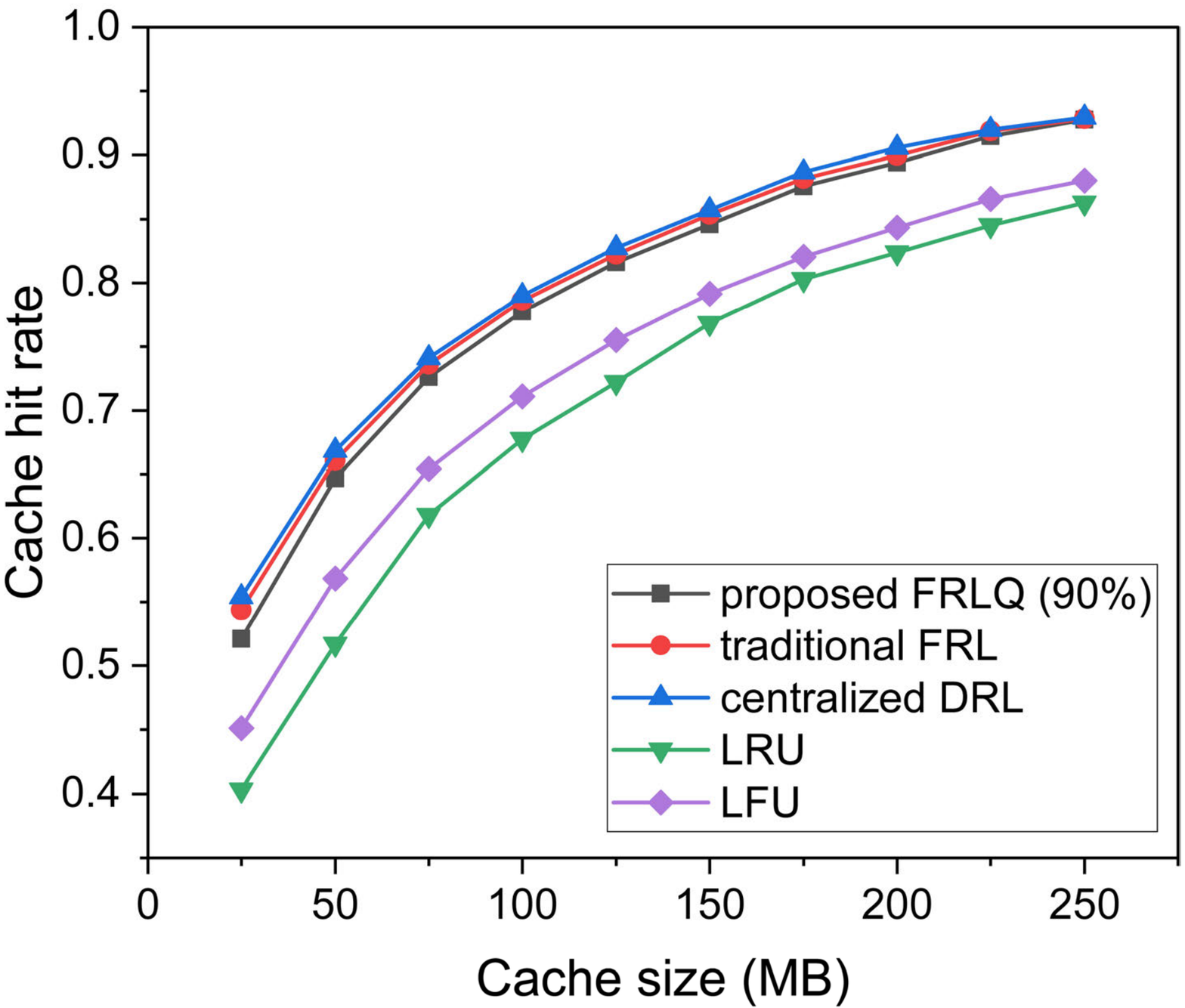}
  \caption{Cache hit rate versus cache size of single F-AP for different caching policies.}\label{hitrate_cachesize}
\end{figure}

As shown in Fig. \ref{hitrate_cachesize}, Fig. \ref{delay_alpha} and Fig. \ref{delay_number}, by considering the cache size, the skewness factor and the number of F-APs, the proposed caching policy outperforms LFU and LRU in terms of cache hit rate and average request delay.
In LFU, the F-AP updates contents according to the requested frequency. In LRU, the F-AP updates contents according to whether the content has been recently requested.
The two caching schemes ignore cooperative caching among multiple F-APs in the network.
For the centralized DRL-based caching policy, all content cache updates are uniformly scheduled by the cloud server. And the users' data will be uploaded to the cloud server, which threatens the data security.
This policy is also difficult to deploy in reality, because the load pressure on backhaul links is relatively large. On the contrary, the proposed policy uploads significantly fewer parameters for the global model update, whereas it achieves almost the same performance as the FRL-based caching policy. Obviously, the proposed caching policy is more lightweight to deploy.

\section{Conclusions}
In this paper, we have proposed the federated reinforcement learning method with quantization for cooperative edge caching in F-RANs. Considering the dynamic network environment, we have formulated the cooperative caching problem as an MDP. We have deployed dueling DQN in each F-AP to find the optimal caching policy by learning the user request behavior and content popularity. We have also established an FL framework between the edge layer and the cloud server.
Local DRL models from multiple F-APs can be trained collaboratively through the FL framework.
On one hand, it avoids the data security problem caused by directly transmitting data of users to the cloud server. On the other hand, it solves the problem of insufficient samples for training a DRL model in single F-AP. Finally, the uploaded local model has been quantized to make the model lighter, which reduces the network transmission pressure.
Furthermore, we have also demonstrated the global convergence and low time complexity of the proposed policy.
Simulation results have shown that the proposed policy achieves almost the same performance as the benchmark schemes, and is more lightweight and easy to deploy.
Future work will explore the joint optimization of user request delay and model transfer delay in F-RANs.

\begin{figure}[t!]
  \centering
  \includegraphics[width=0.42\textwidth]{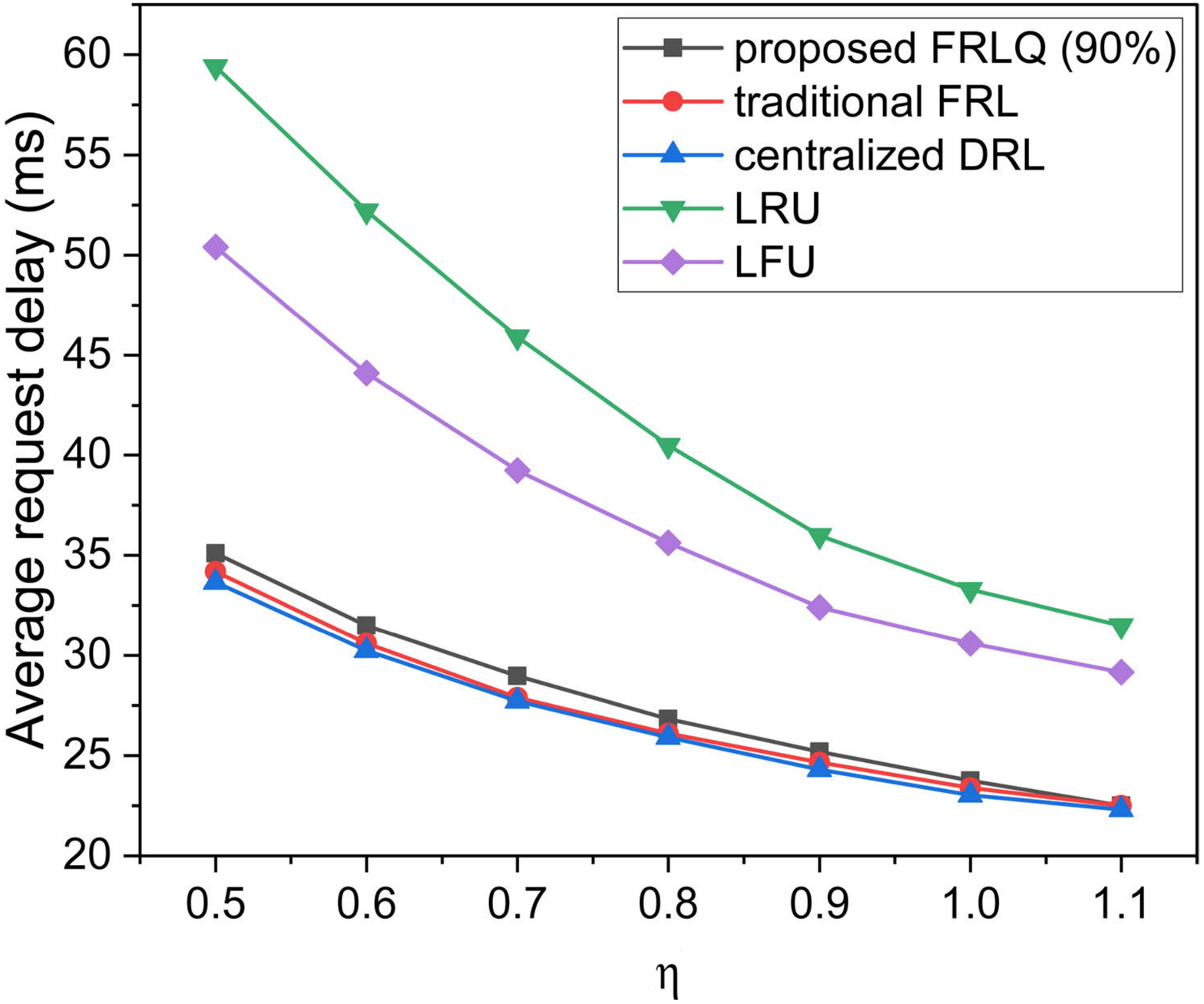}
  \caption{Average request delay versus the skewness factor for different caching policies.}\label{delay_alpha}
\end{figure}

\begin{figure}[t!]
  \centering
  \includegraphics[width=0.42\textwidth]{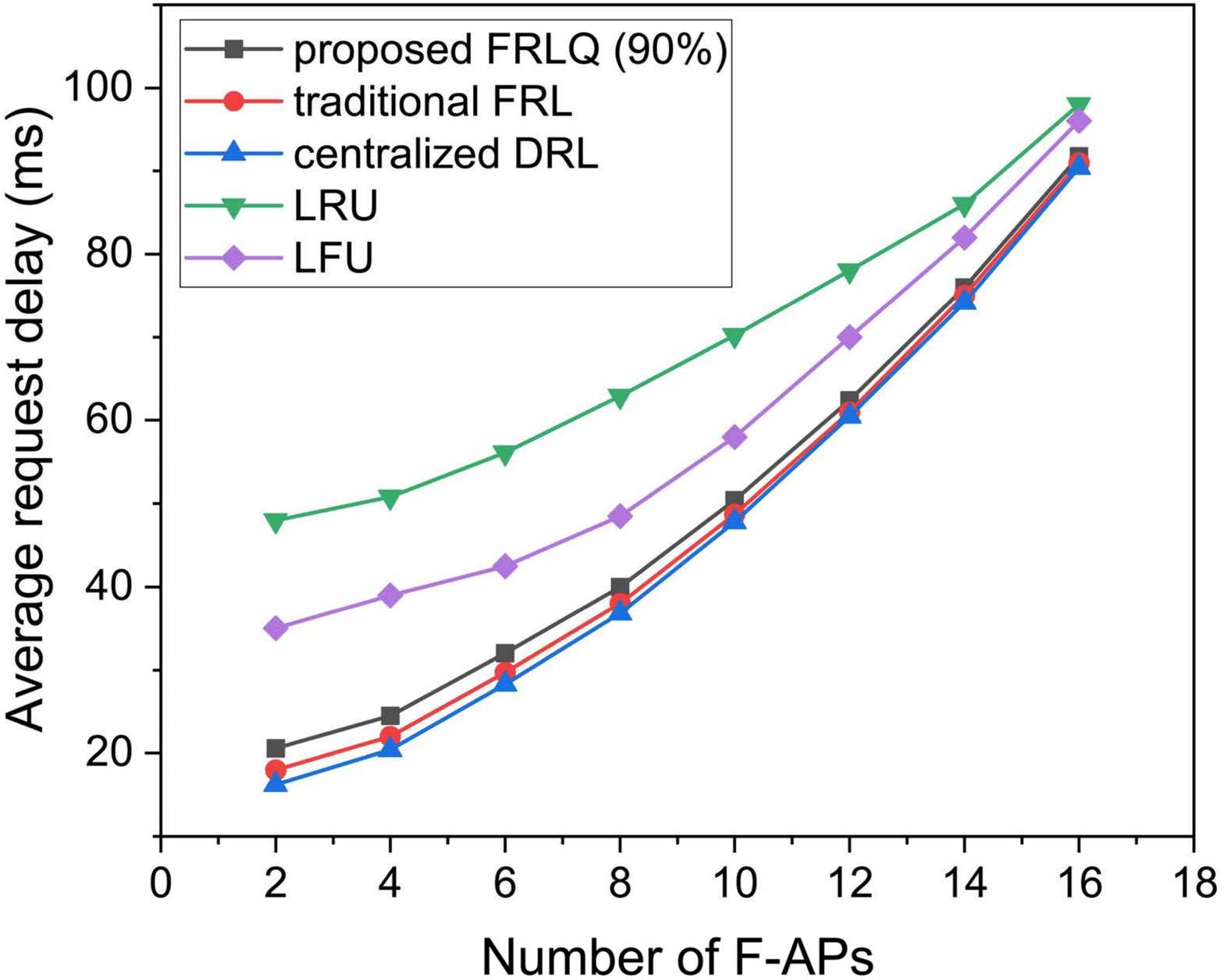}
  \caption{Average request delay versus number of F-APs for different caching policies.}\label{delay_number}
\end{figure}

\begin{appendices}
\section{Proof of Theorem 1}
To prove the convergence of problem (\ref{Theglobalmodelproblem}), i.e., the parameter ${\bar{\theta }}_{t}$ is convergent, it is necessary to prove that $\mathbb{E}{{\left\| {{{\bar{v }}}_{t+1}}-{{\theta }^{*}} \right\|}^2}\le (1-\mu {{\varphi }_{t}})\mathbb{E}{{\|{{\bar{\theta }}_{t}}-{{\theta }^{*}}\|}}^2 +{\varphi }_{t}^{2}H$, i.e., the distance between the iterative parameter ${\bar{\theta }}_{t}$ and the optimal parameter $\theta^*$ is bounded.
Note that ${{\bar{v}}_{t+1}}={{\bar{\theta }}_{t}}-{{\varphi }_{t}}{{g}_{t}}$. Then, we have:
\begin{equation}\label{function1}\begin{aligned}
  {{\left\| {{{\bar{v}}}_{t+1}}-{{\theta}^{*}} \right\|}^{2}}=&{{\left\| {{{\bar{\theta}}}_{t}}-{{\varphi }_{t}}{{g}_{t}}-{{\theta}^{*}}-{{\varphi }_{t}}{{{\bar{g}}}_{t}}+{{\varphi }_{t}}{{{\bar{g}}}_{t}} \right\|}^{2}} \\
  =& {{\left\| {{{\bar{\theta}}}_{t}}-{{\varphi }_{t}}{{{\bar{g}}}_{t}}-{{\theta}^{*}} \right\|}^{2}}+{\varphi }_{t}^{2}{{\left\| {{{\bar{g}}}_{t}}-{{g}_{t}} \right\|}^{2}}\\
  &+2{{\varphi }_{t}}\langle {{{\bar{\theta}}}_{t}}-{{\varphi }_{t}}{{{\bar{g}}}_{t}}-{{\theta}^{*}},-{{g}_{t}}+{{{\bar{g}}}_{t}}\rangle.\\
\end{aligned}
\end{equation}

For ${{\left\| {{{\bar{\theta}}}_{t}}-{{\varphi }_{t}}{{{\bar{g}}}_{t}}-{{\theta}^{*}} \right\|}^{2}}$ in (\ref{function1}), we can further establish:
\begin{equation}\label{function2}
  {{\left\| {{{\bar{\theta }}}_{t}}-{{\varphi }_{t}}{{{\bar{g}}}_{t}}-{{\theta }^{*}} \right\|}^{2}}={{\left\| {{{\bar{\theta }}}_{t}}-{{\theta }^{*}} \right\|}^{2}}-2{{\varphi }_{t}}\langle {{\bar{\theta }}_{t}}-{{\theta }^{*}},{{\bar{g}}_{t}}\rangle +{\varphi }_{t}^{2}{{\left\| {{{\bar{g}}}_{t}} \right\|}^{2}}.
\end{equation}
Considering that $f_n(\theta)$ is $\beta$-smooth in Assumption 1, we have:
\begin{equation}\label{function3}
  {{\left\| \nabla {{f}_{n}}(\theta _{t}^{n}) \right\|}^{2}}\le 2\beta({{f}_{n}}(\theta _{t}^{n}-f_{n}^{*}(\theta ))).
\end{equation}
Since the two-norm is convex, combined with (\ref{function3}), we have:
\begin{equation}\label{function4}
  \begin{aligned}
{\varphi }_{t}^{2}{{\left\| {{{\bar{g}}}_{t}} \right\|}^{2}} &\le {\varphi }_{t}^{2}\sum{{{p}_{n}}}{{\left\| \nabla {{f}_{n}}(\theta _{t}^{n}) \right\|}^{2}}\\
&\le 2\beta{\varphi }_{t}^{2}\sum{{{p}_{n}}({{f}_{n}}(\theta _{t}^{n}-f_{n}^{*}(\theta ))}. \\
\end{aligned}
\end{equation}
For $-2{{\varphi }_{t}}\langle {{{\bar{\theta }}}_{t}}-{{\theta }^{*}},{{{\bar{g}}}_{t}}\rangle$ in (41), we can further establish:
\begin{equation}\label{function5}
  \begin{aligned}
  -2{{\varphi }_{t}}\langle {{{\bar{\theta }}}_{t}}-{{\theta }^{*}},{{{\bar{g}}}_{t}}\rangle =&-2{{\varphi }_{t}}\sum{{{p}_{n}}\langle {{{\bar{\theta }}}_{t}}}-{{\theta }^{*}},\nabla {{f}_{n}}(\theta _{t}^{n})\rangle  \\
  =&-2{{\varphi }_{t}}\sum{{{p}_{n}}\langle {{{\bar{\theta }}}_{t}}}-\theta _{t}^{n},\nabla {{f}_{n}}(\theta _{t}^{n})\rangle \\
 &-2{{\varphi }_{t}}\sum{{{p}_{n}}\langle \theta _{t}^{n}-{{\theta }^{*}}},\nabla {{f}_{n}}(\theta _{t}^{n})\rangle.  \\
\end{aligned}
\end{equation}
According to the Cauchy-Schwarz inequality and the matrix inequality, we have:
\begin{equation}\label{function6}
  -2\langle {{\bar{\theta }}_{t}}-\theta _{t}^{n},\nabla {{f}_{n}}(\theta _{t}^{n})\rangle \le \frac{1}{{{\varphi }_{t}}}{{\left\| {{{\bar{\theta }}}_{t}}-\theta _{t}^{n} \right\|}^{2}}+{{\varphi }_{t}}{{\left\| \nabla {{f}_{n}}(\theta _{t}^{n}) \right\|}^{2}}.
\end{equation}
Considering that $f_n(\theta)$ is $\mu$-strongly convex in Assumption 1, we have:
\begin{equation}\label{function7}
  -\langle \theta _{t}^{n}-{{\theta }^{*}},\nabla {{f}_{n}}(\theta _{t}^{n})\rangle \le -({{f}_{n}}(\theta _{t}^{n})-{{f}_{n}}(\theta{{}^{*}}))-\frac{\mu }{2}{{\left\| \theta _{t}^{n}-{{\theta }^{*}} \right\|}^{2}}.
\end{equation}
Therefore, combined with (\ref{function4}), (\ref{function5}), (\ref{function6}) and (\ref{function7}), the formula (\ref{function2}) can be rewritten as:
\begin{equation}\label{function8}
\begin{aligned}
   &{\left\| {{{\bar{\theta }}}_{t}}-{{\varphi }_{t}}{{{\bar{g}}}_{t}}-{{\theta }^{*}} \right\|}^{2}\\
   &\le  {{\left\| {{{\bar{\theta }}}_{t}}-{{\theta }^{*}} \right\|}^{2}}
   +2\beta{\varphi }_{t}^{2}\sum{p_n({{f}_{n}}(\theta _{t}^{n}-f_{n}^{*}(\theta )))}\\
  &+{{\varphi }_{t}}\sum{{{p}_{n}}(\frac{1}{{{\varphi }_{t}}}}{{\left\| {{{\bar{\theta }}}_{n}}-\theta _{t}^{n} \right\|}^{2}}+{{\varphi }_{t}}{{\left\| \nabla {{f}_{n}}(\theta _{t}^{n}) \right\|}^{2}}) \\
  &-2{{\varphi }_{t}}\sum{{{p}_{n}}({{f}_{n}}(\theta _{t}^{n})-{{f}_{n}}({{\theta }^{*}})-\frac{\mu }{2}}{{\left\| \theta _{t}^{n}-{{\theta }^{*}} \right\|}^{2}}) \\
 &\le (1-\mu {{\varphi }_{t}}){{\left\| {{{\bar{\theta }}}_{t}}-{{\theta }^{*}} \right\|}^{2}}+{{\sum{{{p}_{n}}\left\| {{{\bar{\theta }}}_{t}}-\theta _{t}^{n} \right\|}}^{2}}\\
 &+4\beta{\varphi }_{t}^{2}\sum{{p}_{n}}({{f}_{n}}(\theta _{t}^{n})-f_{n}^{*}(\theta ))\\
 &-2{{\varphi }_{t}}\sum{{{p}_{n}}({{f}_{n}}(\theta _{t}^{n})-{{f}_{n}}({{\theta }^{*}}))}.
 \end{aligned}
\end{equation}

We set ${{\gamma }_{t}}=2{{\varphi }_{t}}(1-2\beta{{\varphi }_{t}})$. Since ${{\varphi }_{t}}\le \frac{1}{4\beta}$, ${{\varphi }_{t}}\le {{\gamma }_{t}}\le 2{{\varphi }_{t}}$ can be obtained. If the data is independent and identically distributed (IID), then $\Phi$ will gradually tend to zero as the number of iterations increases. If the data is non-IID, then $\Phi$ is not equal to zero. Then, we have:
  \begin{align}\label{function9}
  & 4\beta{\varphi }_{t}^{2}\sum{{{p}_{n}}({{f}_{n}}(\theta _{t}^{n})-f_{n}^{*}(\theta ))-2{{\varphi }_{t}}\sum{{{p}_{n}}({{f}_{n}}(\theta _{t}^{n})-{{f}_{n}}({{\theta }^{*}}}})) \notag \\
 & =-2{{\varphi }_{t}}(1-2\beta{{\varphi }_{t}})\sum{{{p}_{n}}}({{f}_{n}}(\theta _{t}^{n})-f_{n}^{*}(\theta )) \notag\\
 &+2{{\varphi }_{t}}\sum{{{p}_{n}}({{f}_{n}}({{\theta }^{*}})-f_{n}^{*}(\theta ))} \notag\\
 & =-{{\gamma }_{t}}\sum{{{p}_{n}}({{f}_{n}}(\theta _{t}^{n})-{{f}}(\theta^* ))+4\beta{\varphi }_{t}^{2}\Phi }.
\end{align}
For $\sum{{{p}_{n}}({{f}_{n}}(\theta _{t}^{n})-{{f}}(\theta^* ))}$ in (\ref{function9}), we can further establish:
\begin{equation}\label{function10}
  \begin{aligned}
  & \sum{{{p}_{n}}({{f}_{n}}(\theta _{t}^{n})-{{f}}(\theta^* ))}\\
 &= \sum{{{{p}_{n}}({{f}_{n}}(\theta _{t}^{n})-}}{{f}_{n}}({{{\bar{\theta }}}_{t}}))+\sum{{{p}_{n}}({{f}_{n}}({{{\bar{\theta }}}_{t}})-{{f}}(\theta^* ))} \\
 & \ge \sum{{{p}_{n}}\langle \nabla }{{f}_{n}}({{{\bar{\theta }}}_{t}}),\theta _{t}^{n}-{{{\bar{\theta }}}_{t}}\rangle +{{f}_{n}}({{{\bar{\theta }}}_{t}})-{{f}}(\theta^* ) \\
 & \ge -\sum{{{p}_{n}}\left[ {{\varphi }_{t}}\beta({{f}_{n}}({{{\bar{\theta }}}_{t}})-f_{n}^{*}(\theta )+\frac{1}{2{\varphi }_{t} }{{\left\| \theta _{t}^{n}-{{{\bar{\theta }}}_{t}} \right\|}^{2}} \right]}\\
 &+\left( f({{{\bar{\theta }}}_{t}})-{{f}}(\theta^* ) \right). \\
\end{aligned}
\end{equation}
In the last inequality of (\ref{function10}), we consider ${{\varphi }_{t}}\beta-1\le -\frac{3}{4}$ and $\sum{{{p}_{n}}({{f}_{n}}({{{\bar{\theta }}}_{t}}})-{{f}}(\theta^* )) = f({{\bar{\theta }}_{t}})-{{f}^{*}}({{\theta }}) \ge 0$. Combined with (\ref{function10}), (\ref{function9}) can be rewritten as:
\begin{equation}\label{function11}
  \begin{aligned}
  & 4\beta{\varphi }_{t}^{2}\sum{{{p}_{n}}({{f}_{n}}(\theta _{t}^{n}})-f_{n}^{*}(\theta ))-2{{\varphi }_{t}}\sum{{{p}_{n}}({{f}_{n}}(\theta _{t}^{n})-{{f}_{n}}({{\theta }^{*}}})) \\
 & \le {{\gamma }_{t}}\sum{{{p}_{n}}\left[ {{\varphi }_{t}}\beta({{f}_{n}}({{{\bar{\theta }}}_{t}})-f_{n}^{*}(\theta ))+\frac{1}{2{\varphi }_{t} }{{\left\| \theta _{t}^{n}-{{{\bar{\theta }}}_{t}} \right\|}^{2}} \right]}\\
 &-\gamma_t (f({{{\bar{\theta }}}_{t}})-{{f}}(\theta^* ))+4\beta{{\varphi }_{t}^2}\Phi  \\
 & ={{\gamma }_{t}}({{\varphi }_{t}}\beta-1)\sum{{{p}_{n}}({{f}_{n}}({{{\bar{\theta }}}_{t}}})-{{f}}(\theta^* ))+(4\beta{\varphi }_{t}^{2}\Phi+{{\gamma }_{t}}{{\varphi }_{t}}\beta)\\
 &+\frac{\gamma }{2{{\varphi }_{t}}}{{\sum{{{p}_{n}}\left\| \theta _{t}^{n}-{{{\bar{\theta }}}_{t}} \right\|}}^{2}}
  \le 6\beta{\varphi }_{t}^{2}\Phi +{{\sum{{{p}_{n}}\left\| \theta _{t}^{n}-{{{\bar{\theta }}}_{t}} \right\|}}^{2}}.
\end{aligned}
\end{equation}
Combined with (\ref{function8}) and (\ref{function11}), (\ref{function2}) can be rewritten as:
  \begin{multline}\label{function12}
  {{\left\| {{{\bar{\theta }}}_{t}}-{{\varphi }_{t}}{{{\bar{g}}}_{t}}-{{\theta }^{*}} \right\|}^{2}} \le (1-\mu {{\varphi }_{t}}){{\left\| {{{\bar{\theta }}}_{t}}-{{\theta }^{*}} \right\|}^{2}}\\
  +6\beta{\varphi }_{t}^{2}\Phi +2\sum{{{p}_{n}}{{\left\| \theta _{t}^{n}-{{{\bar{\theta }}}_{t}} \right\|}^{2}}}.
   \end{multline}

According to Assumption 2, the expected squared norm $G$ of the stochastic gradient is bounded, namely:
\begin{equation}\label{function14}
  \begin{aligned}
   &{{\sum{{{p}_{n}}\left\| {{{\bar{\theta }}}_{t}}-\theta _{t}^{n} \right\|}}^{2}}\\&={{\sum{{{p}_{n}}\left\| \theta _{t}^{n}-{{{\bar{\theta }}}_{{{t}_{0}}}}-({{{\bar{\theta }}}_{t}}-{{{\bar{\theta }}}_{{{t}_{0}}}}) \right\|}}^{2}} \\
  &\le {{\sum{{{p}_{n}}\left\| \theta _{t}^{n}-{{{\bar{\theta }}}_{{{t}_{0}}}} \right\|}}^{2}} \\
  &\le \sum{{{p}_{n}}\sum\nolimits_{t={{t}_{0}}}^{t-1}{(X-1){\varphi }_{t}^{2}}}\left\| \nabla {{f}_{n}}(\theta _{t}^{n},\xi _{t}^{n}) \right\|^2 \\
  &\le 4{\varphi }_{t}^{2}{{(X-1)}^{2}}{{G}^{2}}.\\
  \end{aligned}
\end{equation}
According to Assumption 3, the gradient variance $\sigma _{n}^{2}$ is bounded, namely:
\begin{equation}\label{function13}
  \begin{aligned}
   {{\left\| {{g}_{t}}-{{{\bar{g}}}_{t}} \right\|}^{2}}&=\left\| \sum{{{p}_{n}}(\nabla {{f}_{n}}({{\theta }^n_{t}},\xi _{t}^{n})-\nabla {{f}_{n}}(\theta _{t}^{n})}) \right\|^2 \\
 & \le \sum{p_{n}^{2}}\sigma _{n}^{2}. \\
\end{aligned}
\end{equation}

For $2{{\varphi }_{t}}\langle {{{\bar{\theta}}}_{t}}-{{\varphi }_{t}}{{{\bar{g}}}_{t}}-{{\theta}^{*}},-{{g}_{t}}+{{{\bar{g}}}_{t}}\rangle$ in (\ref{function1}), since $\mathbb{E}({{{\bar{g}}}_{t}}-{{g}_{t}})=0$, we have:
\begin{equation}\label{function15}
  \mathbb{E}(2{{\varphi }_{t}}\langle {{{\bar{\theta}}}_{t}}-{{\varphi }_{t}}{{{\bar{g}}}_{t}}-{{\theta}^{*}},-{{g}_{t}}+{{{\bar{g}}}_{t}}\rangle)=0.
\end{equation}

Therefore, by combining (\ref{function1}), (\ref{function12}), (\ref{function13}), (\ref{function14}) and (\ref{function15}), the following relationship can be readily established:
\begin{multline}\label{function16}
  \mathbb{E}{{\left\| {{{\bar{v }}}_{t+1}}-{{\theta }^{*}} \right\|}^2}\le (1-\mu {{\varphi }_{t}})\mathbb{E}{{\|{{\bar{\theta }}_{t}}-{{\theta }^{*}}\|}}^2+{\varphi }_{t}^{2}H.
\end{multline}

This completes the proof.

\section{Proof of Theorem 2}
We prove (\ref{decay}) by mathematical induction. The definition of $\rho$ guarantees that (\ref{decay}) holds for $t\ge1 $. Then, we have:
\begin{equation}\label{zhujianjianxiao}
\begin{aligned}
  {{\Delta }_{t+1}}&\le (1-{{\varphi }_{t}}\mu ){{\Delta }_{t}}+{{\varphi }_{t}^2}H \\
 & \le (1-\frac{b\mu }{t+a})\frac{\rho }{t+a}+\frac{{{b}^{2}}H}{{{(t+a)}^{2}}} \\
 & =\frac{t+a-1}{{{(t+a)}^{2}}}\rho +\left[\frac{{{b}^{2}}H}{{{(t+a)}^{2}}}-\frac{b\mu -1}{{{(t+a)}^{2}}}\rho\right] \\
 & \le \frac{\rho }{t+a+1}. \\
\end{aligned}
\end{equation}
It indicates that ${\Delta }_{t}$ gradually decreases with time.

This completes the proof.
 \end{appendices}

\bibliographystyle{IEEEtran}      
\bibliography{manuscript}

\end{document}